\documentclass[11pt]{article}
\usepackage[margin=1in]{geometry}

\usepackage{amsmath, amsthm, amssymb, enumerate, hyperref, algorithm}
\newtheorem{theorem}{Theorem}[section]

\newtheorem{lemma}[theorem]{Lemma}
\newtheorem{claim}[theorem]{Claim}
\newtheorem{remark}[theorem]{Remark}

\newtheorem{oracle}{Oracle}

\newcommand{\R}{\mathbb{R}}
\newcommand{\ole}{\overleftrightarrow}
\newcommand{\obe}{\overleftarrow}

\title{\vspace{-15pt}A Strongly Polynomial Algorithm for a Class of Minimum-Cost Flow Problems with
  Separable Convex Objectives}

\author{L\'aszl\'o A. V\'egh
\thanks{Department of Mathematics, London School of Economics.
E-mail: {\tt L.Vegh\char'100 lse.ac.uk}.
This work was done at the 
College of Computing, Georgia Institute of Technology,
supported by NSF Grant CCF-0914732.
This is a  revision of the work published in 
STOC'12, May 19 - 22 2012, New York, NY, USA
Copyright 2012 ACM 978-1-4503-1245-5/12/05
 }
}
\begin{document}
\maketitle
\begin{abstract}
A well-studied nonlinear extension of the minimum-cost flow problem is 
to minimize the objective $\sum_{ij\in E} C_{ij}(f_{ij})$ over feasible flows $f$, where on every arc $ij$ of the network,
$C_{ij}$ is a convex function.  We give a strongly polynomial
algorithm for the case when all $C_{ij}$'s are convex quadratic
functions, settling an open problem raised e.g. by Hochbaum
\cite{Hochbaum94}. We also give strongly polynomial algorithms for
computing market equilibria in Fisher markets with linear utilities and with spending constraint utilities, that can be formulated in this framework (see Shmyrev \cite{Shmyrev09}, Devanur et al. \cite{Birnbaum11}). For the latter class this resolves an open question raised by Vazirani \cite{Vazirani10spending}.
The running time is $O(m^4\log m)$ for quadratic costs,  $O(n^4+n^2(m+n\log n)\log n)$ for Fisher's markets with linear utilities and 
$O(mn^3 +m^2(m+n\log n)\log m)$ for spending constraint utilities.

All these algorithms are presented in a common framework that addresses the general problem setting.
Whereas it is impossible to give a strongly polynomial algorithm for the general problem even in an approximate sense (see Hochbaum \cite{Hochbaum94}),
we show that assuming the existence of certain black-box oracles, one can give an algorithm using a strongly polynomial number of arithmetic operations and oracle calls only. The particular algorithms can be derived by implementing these oracles in the respective settings.
\end{abstract}

%
%
%



\section{Introduction}

Let us consider an optimization problem where the input is given by $N$ numbers.
An algorithm for such a problem is called {\em  strongly polynomial}
(see \cite{gls}),
if
{\em(i)} it uses only elementary arithmetic operations (addition, subtraction, multiplication, division, and comparison); {\em (ii)} the
number of these operations is bounded by a polynomial of $N$ {\em
  (iii)}
if all numbers in the input are rational, then all numbers occurring in
the computations are rational numbers of size polynomially bounded in 
$N$ and the maximum size of the input numbers.
Here,  the size of a rational number $p/q$ is defined as
$\lceil \log_2 (p+1)\rceil+\lceil \log_2 (q+1)\rceil$.

The flow polyhedron is defined on a directed network
$G=(V,E)$ by arc capacity and node demand constraints; throughout the
paper, $n=|V|$ and $m=|E|$. 
We study the {\em minimum cost
  convex separable flow problem}: for feasible
flows $f$, the objective is to minimize $\sum_{ij\in E} C_{ij}(f_{ij})$,
where on each arc $ij\in E$, $C_{ij}$ is a differentiable convex function.
We give a strongly polynomial algorithm for the case of convex
quadratic functions, i.e. if
$C_{ij}(\alpha)=c_{ij}\alpha^2+d_{ij}\alpha$ with $c_{ij}\ge 0$ for
every arc $ij\in E$. We also give strongly polynomial algorithms for
Fisher's market with linear utilities and with spending constraint utilities;
these problems can be formulated as minimum cost convex separable flow
problems, as shown respectively by Shmyrev \cite{Shmyrev09} and by Devanur et al. \cite{Birnbaum11}.
The formulations involve linear cost functions and the function $\alpha(\log \alpha -1)$ on certain arcs.

These algorithms are obtained as special implementations of an algorithm that
works for the general problem setting under certain assumptions. 
We assume that the functions are represented by oracles (the specific details are provided later), and two further black-box
oracles are provided. We give a strongly polynomial algorithm in the sense that
it uses only basic arithmetic operations and oracle
calls; the total number of these operations is
polynomial in $n$ and $m$. We then verify  our assumptions 
for convex quadratic objectives and the Fisher markets, and show that
we can obtain strongly polynomial algorithms for these problems.

Flows with separable convex objectives are natural convex extensions of minimum-cost flows with several
applications such as matrix balancing or traffic networks, see
\cite [Chapter 14]{amo} for further references. Polynomial-time
combinatorial algorithms were given 
by Minoux \cite{Minoux86} in 1986, 
by Hochbaum and Shantikumar
\cite{Hochbaum90} in 1990, and by Karzanov and
McCormick \cite{Karzanov97} in 1997. The latter two approaches are able to 
solve even more general problems of minimizing a separable (not
necessarily differentiable) convex objective over a polytope given by a
matrix with a bound on its largest subdeterminant. Both approaches
give polynomial, yet not strongly polynomial algorithms.

In contrast, for the same problems with linear objectives, Tardos
\cite{Tardos85,Tardos86} gave strongly polynomial algorithms. 
One might wonder
whether this could also be extended to the convex setting.
This seems impossible for arbitrary convex objectives by the very
nature of the problem: the optimal solution might be irrational, and
thus the exact optimum cannot be achieved.

Beyond irrationality, the result of Hochbaum \cite{Hochbaum94} shows that
it is impossible to find an $\varepsilon$-accurate solution%
\footnote{A solution $x$ is called $\varepsilon$-accurate if there exists
  an optimal solution $x^*$ with $||x-x^*||_\infty\le\varepsilon$.}
 in strongly polynomial
time even for a network consisting of parallel arcs between a source
and a 
sink node and the
$C_{ij}$'s being polynomials of degree at least three. This is based
on Renegar's  result  \cite{Renegar87} showing the impossibility of  finding
$\varepsilon$-approximate roots of polynomials in strongly polynomial time.
This is an unconditional impossibility result in a
computation model allowing basic arithmetic operations and comparisons;
it does not rely on
any complexity theory assumptions.

The remaining  class of polynomial objectives with hope of strongly
polynomial algorithms is where every cost function is  convex
quadratic.
If all coefficients are rational, then  the existence of a rational optimal solution is guaranteed. 
Granot and Skorin-Kapov \cite{Granot90}
 extended Tardos's method
\cite{Tardos86} to solving separable convex quadratic optimization
problems with linear constraints, where the running time depends only on the entries of the
constraint matrix and the coefficients of the quadratic terms in the
objective. However, this algorithm is not strongly polynomial because of the dependence on the quadratic terms.

The existence of a strongly polynomial algorithm for the quadratic
flow problem thus remained an important open question (mentioned e.g.\
in \cite{Hochbaum94, Cosares94, Hochbaum07, Granot90,Tamir93}).
The survey paper \cite{Hochbaum07} gives an overview of special cases solvable in strongly polynomial time. These include a fixed number of suppliers (Cosares and Hochbaum, \cite{Cosares94}), and series-parallel graphs (Tamir \cite{Tamir93}).
We resolve this question affirmatively, providing a strongly
polynomial algorithm for the general problem in time $O(m^4\log m)$.

There is an analogous situation for convex closure sets: \cite{Hochbaum94} shows that no strongly polynomial algorithm may exist in general,
but for quadratic cost functions,  Hochbaum and Queyranne
 \cite{Hochbaum03} gave a strongly polynomial algorithm.

\medskip

An entirely different motivation of our study comes from the study of
market equilibrium algorithms. Devanur et al.
\cite{Devanur08} developed a polynomial time combinatorial algorithm for a classical problem in economics, Fisher's market with linear utilities. 
This motivated a line of research to develop combinatorial algorithms
for other market equilibrium problems. For a survey, see \cite[Chapter
5]{Nisan07} or \cite{Vazirani11}.  All these problems are described by
rational convex programs (see \cite{Vazirani11}).
For the linear Fisher market problem, a strongly polynomial algorithm was given by Orlin \cite{Orlin10}.

To the extent of the author's knowledge, these rational convex programs have been
considered so far as a new domain in combinatorial
optimization. An explicit connection to classical flow problems  was pointed out in the recent paper \cite{Vegh11}. It turns out that the linear Fisher market problem, along with several other problems, is captured by a concave extension of the classical 
generalized flow problem,  solvable by a polynomial time combinatorial algorithm.

 The paper \cite{Vegh11} uses the convex programming formulation of linear Fisher markets by Eisenberg and Gale \cite{Eisenberg59}. An alternative convex program for the same problem was given by
 Shmyrev \cite{Shmyrev09}. This formulation turns out to be  a convex separable minimum-cost flow problem.
Consequently, equilibrium for linear Fisher market  can be computed by the general algorithms 
\cite{Hochbaum90,Karzanov97} (with a final
transformation of a close enough approximate solution to an exact
optimal one).

The class of convex flow problems solved in this paper also contains
the formulation of Shmyrev, yielding an alternative strongly
polynomial algorithm for  linear Fisher market.
  Devanur et al.\ \cite{Birnbaum11} gave an analogous formulation for Fisher's
  market with spending constraint utilities, defined by Vazirani
\cite{Vazirani10spending}.  
For this problem, we obtain the 
 first strongly polynomial algorithm.
Our running time bounds are $O(n^4+n^2(m+n\log
n)\log n)$ for linear and $O(mn^3 + m^2(m+n\log
n)\log m)$ for spending constraint utilities, with $m$ being
the number of segments in the latter problem.  For the linear case,
 Orlin \cite{Orlin10} used the assumption $m=O(n^2)$
and achieved  running time $O(n^4\log n)$, the same as ours under this
 assumption. So far, no extensions of \cite{Orlin10} are known for
 other market settings.

\subsection{Prior work}

For linear minimum-cost flows, the first polynomial time algorithm was
the scaling method by  Edmonds and Karp \cite{Edmonds72}. The 
current most efficient strongly polynomial algorithm, given by Orlin
\cite{Orlin93}, is also based on this framework. On the other hand,
Minoux extended  \cite{Edmonds72} to 
the convex minimum-cost flow problem, first 
to convex quadratic flows \cite{Minoux84}, later to general
convex objectives \cite{Minoux86}.
Our algorithm is an enhanced version of the latter algorithm, in the
spirit of Orlin's technique \cite{Orlin93}. However, there are 
important differences that make the nonlinear setting significantly
harder. Let us remark that Orlin's strongly polynomial algorithm for
linear Fisher market \cite{Orlin10} is also based on the ideas of \cite{Orlin93}.
In what follows, we give an
informal overview of the key ideas of these algorithms that motivated
our result. For more detailed references and proofs, we refer the reader to \cite{amo}.

The algorithm of Edmonds and Karp consists of $\Delta$-phases for a
scaling parameter $\Delta$. Initially, $\Delta$ is set to a large
value, and decreases by at least  a factor of two at the end of each
phase. An optimal solution can be obtained for sufficently small
$\Delta$. 
The elementary step of the $\Delta$-phase transports
$\Delta$ units of flow from a node with excess at least $\Delta$ to
another node with demand at least $\Delta$. This is done on a shortest path
in the $\Delta$-residual network, the graph of residual arcs with
capacity at least $\Delta$. An invariant property maintained in the
$\Delta$-phase is that the $\Delta$-residual network does not contain
any negative cost cycles. When moving to the next phase, the flow on the
arcs has to be slightly modified to restore the invariant property.

Orlin's algorithm (\cite{Orlin93}, see also \cite[Chapter 10.6]{amo}) works on a problem instance with no
upper capacities on the arcs (every minimum-cost flow problem can be
easily transformed to this form).
The basic idea is that  if
the algorithm runs for infinite number of phases, then the solution converges to an optimal
solution; furthermore, the total change of the flow value in the $\Delta$-phase
and all subsequent phases is at most $4n\Delta$ on every arc.
Consequently, if an arc $ij$ has flow
 $>4n\Delta$ in the $\Delta$-phase, then the flow on $ij$ must
be positive in some optimal solution. 
 Using primal-dual slackness,
this means that $ij$ must be tight for an arbitrary dual optimal
solution (that is, the corresponding dual inequality must hold with equality).  It is shown that within $O(\log n)$ scaling phases, an
 arc $ij$ with flow larger than $4n\Delta$  appears. 

Based on this fact, \cite{Orlin93} obtains the following simple
algorithm. Let us start running the Edmonds-Karp algorithm on the
input graph. Once there is an  arc with flow larger than $4n\Delta$,
it is contracted and the Ed\-monds-Karp algorithm is restarted on the
smaller graph. The method is iterated until the graph reduces to a
single node. A dual optimal solution on the contracted graph can be easily extended
to a dual optimal solution in the original graph by reversing the
contraction operations.
Provided a dual optimal solution, a primal optimal solution can be obtained by a single maximum
flow computation.
The paper \cite{Orlin93} (see also
\cite[Chapter 10.7]{amo}) also contains a second, more efficient algorithm.
When an arc with ``large'' flow is found, instead of contracting and restarting, 
the arc is added to a special forest $F$. The scaling algorithm
exploits properties of this forest and can thereby ensure that a 
new arc enters $F$ in
$O(\log n)$ phases. The running time can be bounded by $O(m\log
n(m+n\log n))$, so far the most efficient minimum-cost flow algorithm known.

\medskip

Let us now turn to the nonlinear setting. 
By the Karush-Kuhn-Tucker (KKT) conditions, a feasible solution is  optimal if and only if
the residual graph contains no
negative cycles with respect to the cost function $C'_{ij}(f_{ij})$.
Minoux's algorithm is a natural extension of the  Edmonds-Karp scaling technique
(see \cite{Minoux84,Minoux86}, \cite[Chapter
14.5]{amo}).
In the $\Delta$-phase  it maintains the invariant that
the $\Delta$-residual graph contains no negative cycle with respect to
the relaxed cost function $(C_{ij}(f_{ij}+\Delta)-C_{ij}(f_{ij}))/\Delta$. When transporting
 $\Delta$-units of flow on a shortest path with respect to this cost
 function, this invariant is maintained. A key observation is that when 
moving to the $\Delta/2$-phase, the invariant can be restored by
changing the flow on each arc by at most $\Delta/2$.
The role of the scaling factor $\Delta$ is twofold: besides being the
quantity of the transported flow, it also approximates optimality in
the following sense.
As $\Delta$ approaches 0, the cost of $ij$ converges to the
 derivative $C'_{ij}(f_{ij})$. Consequently, the solution converges to
 a feasible optimal solution.
 A variant of this algorithm is outlined in Section~\ref{sec:basic}.

\subsection{Overview of  the algorithm for convex quadratic flows}\label{sec:outline-strong}
To formulate the exact assumptions needed for the general algorithm,
several notions have to be introduced. Therefore we postpone the
formulation of our main
result Theorem~\ref{thm:main} to Section~\ref{sec:assump}. Now we
exhibit the main ideas on the example of convex quadratic functions.
We only give an informal overview here without providing all technical
details; the precise definitions and descriptions are given in the
later parts of the paper. Then in Section~\ref{sec:quad}, we show how
the general framework can  be adapted to convex quadratic functions.

Let us assume that  $C_{ij}(\alpha)=c_{ij}\alpha^2+d_{ij}\alpha$ with
$c_{ij}>0$ for every arc
$ij\in E$, and therefore all cost functions are strictly convex.
This guarantees that the optimal solution is unique.
This assumption is made only for the sake of this overview, and not used in the 
formal presentation starting in Section~\ref{sec:preliminaries}.
However, it is useful as the uniqueness of the optimum enables certain
technical simplifications. We discuss these simplifications at the end
of the section.
Our
problem can be formulated as follows.
\begin{align*}
\min~\sum_{ij\in E}c_{ij}f^2_{ij}+&d_{ij}f_{ij}\nonumber\\
\sum_{j:ji\in E}f_{ji}-\sum_{j:ij\in E}f_{ij}&=b_i\quad\forall i\in V\\
f&\ge 0 
\end{align*}
In a more general formulation, one could have arbitrary upper and
lower capacities on the arcs. However, this can be reduced to the
above form, see
Section~\ref{sec:preliminaries}. 

Let $f^*$ be the optimal solution; it is unique by the strict
convexity of the objective. Let $F^*$ denote the support of $f^*$. 
An optimal solution can be characterized using the Karush-Kuhn-Tucker
conditions: for Lagrange multipliers $\pi:V\to \R$, we have
$\pi_j-\pi_i\le 2c_{ij}f^*_{ij}+d_{ij}$ with equality whenever
$f^*_{ij}>0$, that is, $ij\in F^*$.
Consequently, if $F^*$ is provided, then we can obtain $f^*$ as the
unique solution to the following system of
 linear equations (see Section~\ref{sec:quad} for
details).
\begin{align}
\pi_j-\pi_i&=2c_{ij}f^*_{ij}+d_{ij} \quad \forall ij\in F^*\notag\\
\sum_{j:ji\in F^*}f^*_{ji}-\sum_{j:ij\in F^*}f^*_{ij}&=b_i\quad\forall i\in V\label{sys:quad}\\
f^*_{ij}&=0\quad \forall ij\in E\setminus F^*\notag
\end{align}
 We will assume the existence of the subroutine {\sc Trial}$(F,\hat b)$ (Oracle~\ref{assump:trial}),
where $F\subseteq E$ is an arbitrary arc set and $\hat b:V\to \R$ such that the sum of the $\hat b_i$ values is 0 in
every undirected connected component of $F$. The subroutine solves the modification of (\ref{sys:quad}) when $F^*$ is substituted by $F$ and $b$ by $\hat b$. The system is feasible under the above assumption on $\hat b$, and a solution can be found in time $O(n^{2.37})$ (see Lemma~\ref{lem:quad-solve}).

Our starting point is a variant of Minoux's nonlinear scaling scheme as described above,
with the only difference that the relaxed cost function is replaced
by $C'_{ij}(f_{ij}+\Delta)$ (see  Section~\ref{sec:basic}).

Following Orlin \cite{Orlin93}, we can identify an arc carrying a ``large'' amount
of flow in
$O(\log n)$ steps.  The required amount, $(2n+m+1)\Delta$ at the end
of the $\Delta$-phase, is large enough that
even if we run the algorithm forever and thereby converge to the
optimal solution $f^*$, this arc must remain positive. Consequently,
it must be contained in $F^*$. However, 
we cannot simply contract such an arc as in \cite{Orlin93}. The reason is that  the KKT-conditions
give $\pi_j-\pi_i=c_{ij}f^*_{ij}+d_{ij}$, a condition containing both
primal and dual (more precisely, Lagrangian) variables simultaneously.

In every phase of the algorithm, we shall maintain a set $F\subseteq
F^*$ of arcs, called {\sl revealed} arcs. $F$ will be extended by a
new arc in every $O(\log n)$ phases; thus we find $F^*$ in $O(m\log
n)$ steps (see Theorem~\ref{thm:main-bound}). 
Given a set $F\subseteq F^*$, we introduce some technical
notions; the precise definitions and detailed discussions are given in
Section~\ref{sec:revealed}. First, we waive the
nonnegativity requirement on the arcs in $F$: a vector $E\to \R$ is called  an {\em $F$-pseudoflow}, 
if $f_{ij}\ge 0$ if $ij\in E\setminus F$ but the arcs in $F$ are unconstrained.

 For an $F$-pseudoflow $f$ and a scaling factor $\Delta>0$, the $(\Delta,F)$-residual graph $E_f^F(\Delta)$ contains
all residual arcs where $f$ can be increased by $\Delta$ so that it
remains an $F$-pseudoflow (that is, all arcs in $E$, and all arcs $ji$
where $ij\in F$, or $ij\in E\setminus F$ and $f_{ij}\ge \Delta$.)
We require that the flow $f$ in this phase satisfies
the {\em $(\Delta,F)$-feasibility} property: the graph $E_f^F(\Delta)$
contains no negative cycles with respect to the cost function $C'_{ij}(f_{ij}+\Delta)$.

Let us now describe our algorithm. We start with $F=\emptyset$ and a sufficiently large
$\Delta$ value so that the initial flow $f\equiv 0$ is
$(\Delta,\emptyset)$-feasible. We run the Minoux-type scaling
algorithm sending flow on shortest paths in the $(\Delta,F)$-residual
graph from nodes with excess at least $\Delta$ to nodes with
deficiency at least $\Delta$. 
If there exist no more such paths, we move to the $\Delta/2$-phase,
after a simple modification step that transforms the flow to
a $(\Delta/2,F)$-feasible one, on the cost of increasing the total
excess by at most $m\Delta/2$ (see subroutine \textsc{Adjust} in Section~\ref{sec:revealed}).
We include in $F$ every edge  with
$f_{ij}>(2n+m+1)\Delta$ at the end of  the $\Delta$-phase.

At the end of each phase  when $F$ is extended, we perform a special
subroutine instead of simply moving to the $\Delta/2$-phase.
First, we compute the
discrepancy $D_F(b)$ defined as follows.
Let  $D_F(b)=\max_{K}|\sum_{i\in K} b_i|$, 
where $K$ ranges over the undirected connected components of $F$. (Note that the requirement on $\hat b$ in the subroutine
{\sc Trial}$(F,\hat b)$ above was $D_F(\hat b)=0$.)
 If the discrepancy $D_F(b)$ is large, then it
can be shown that $F$ will be extended within $O(\log n)$ phases as in
Orlin's algorithm (see the first part of the proof of Theorem~\ref{thm:main-bound}). 

If the discrepancy is small, the procedure \textsc{Trial-and-Error} is
performed, consisting of two subroutines.
First, we run the subroutine \textsc{Trial}$(F,\hat b)$, where $\hat b$ is
a small modification of $b$ satisfying $D_F(\hat b)=0$. This returns an
$F$-pseudoflow $\hat f$, satisfying (\ref{sys:quad}) with $F$ in the
place of $F^*$.
(This step be seen as ``pretending'' that $F=F^*$ and trying to
compute an optimal solution under this hypothesis.)
The resulting $\hat f$ is optimal if and only if $F=F^*$.
 Otherwise, 
we use a second subroutine \textsc{Error}$(\hat f,F)$ (see Oracle~\ref{assump:error}), that
returns the smallest value $\hat\Delta>0$ such that $\hat f$ is
$(F,\hat\Delta)$-feasible. This subroutine can be reduced to a minimum
cost-to-time ratio cycle problem (also known as the tramp streamer
problem), see \cite[Chapter 5.7]{amo}; a strongly polynomial time
algorithm was given by Megiddo \cite{Megiddo79}.

If $\hat \Delta<\Delta/2$, then we set $\hat
\Delta$ as our next scaling value and $f=\hat f$ as the next
$F$-pseudoflow - we can proceed since $\hat f$ is $(F,\hat\Delta)$-feasible.
Otherwise, the standard transition to phase $\Delta/2$ is done with
keeping the same flow $f$.
The analysis shows that a new arc shall be revealed in every $O(\log
n)$ phases. The key Lemma~\ref{lem:error-bound} relies on the proximity of $f$ and $\hat
f$, which implies that  \textsc{Trial-and-Error} cannot return the same $\hat f$ if
performed again after $O(\log n)$ phases. Consequently, the set $F$
cannot be the same, and has been therefore extended. Since $|F|\le m$,
this shows that the total number of
scaling phases is $O(m\log n)$. 

\medskip

Besides the impossibility of contraction, an important difference as
compared ot Orlin's algorithm  is that $F$ cannot be assumed to be a
forest (in the
undirected sense). There are simple quadratic instances with the support of an optimal
solution containing cycles. In Orlin's algorithm, progress is always
made by connecting two components of $F$. This will also be an
important event in our algorithm, but sometimes $F$ shall be extended with
arcs inside a component.

\medskip
The main difference when applied to Fisher markets instead of
quadratic costs is the implementation of the black boxes
\textsc{Trial} and \textsc{Error}. These are implemented by a simple linear time algorithm and
the Floyd-Warshall algorithm, respectively.
The description above made the
simplifying assumption that $c_{ij}>0$ for all $ij\in E$, that is, all cost functions are strictly
convex, and thus there is a unique optimal solution. This might not be true even for quadratic costs if $c_{ij}=0$ is allowed on certain arcs. 
 An important difference between the description and the general algorithm is that
in the general algorithm, the set $F^*$ has to be more carefully defined; in particular, it will
contain the support of every optimal solution. We therefore have to
introduce the additional notion of $F$-optimal solutions for $F\subseteq F^*$. The algorithm will find $F$-optimal solutions instead of optimal ones; however, an $F$-optimal solution can be converted to
an optimal solution via an
additional maximum flow subroutine.

\smallskip

The rest of the paper is organized as follows. Section~\ref{sec:preliminaries}
contains the basic definitions and notations. Section~\ref{sec:basic}
presents the simple adaptation of the Edmonds-Karp algorithm for
convex cost functions, following  Minoux \cite{Minoux86}. Our
algorithm in Section~\ref{sec:enhanced} is
built on this algorithm with the addition of the subroutine
\textsc{Trial-and-Error}, that guarantees strongly polynomial running time.
Analysis is given in Section~\ref{sec:analysis}.
Section~\ref{sec:applications}  adapts  the general
algorithm for quadratic utilities, and for Fisher's market with linear
and with spending constraint utilities.
 Section~\ref{sec:discussion} contains a final discussion of the
 results and some open questions.
An Appendix contains the description of the shortest path subroutines used.
A table summarizing notation and concepts can be found
at the end of the paper.

\section{Preliminaries}\label{sec:preliminaries}

Let $G=(V,E)$ be a directed graph, and let $n=|V|$, $m=|E|$. For
notational convenience, we assume that the graph contains no parallel
arcs and no pairs of oppositely directed arcs. Consequently, we can
denote the arc from node $i$ to node $j$ by $ij$. All results
straightforwardly extend to general graphs.
We are given node demands $b:V\rightarrow \R$ with $\sum_{i\in V}b_i=0$.
The flow is restricted to be nonnegative on every arc, but there are no upper capacities.
 On each arc $ij\in E$,
$C_{ij}: \R \rightarrow \R\cup\{\infty\}$ is a convex
function. We allow two types of arcs $ij$:
\begin{itemize}
\item {\sl Free arcs:} $C_{ij}$ is differentiable everywhere on $\R$.
\item {\sl Restricted arcs:} $C_{ij}(\alpha)=\infty$ if $\alpha<0$,
$C_{ij}$ is differentiable on $(0,\infty)$ and
has a left derivative in $0$ that equals $-\infty$; let $C'_{ij}(0)=-\infty$ denote this left derivative.
Let us use the convention  $C'_{ij}(\alpha)=-\infty$ for $\alpha<0$.
\end{itemize}
By convexity, $C'_{ij}$ is continuous on $\R$ for free and on $[\ell_{ij},\infty)$ for restricted arcs.
Restricted arcs will play a role in the Fisher market applications, where the function 
$C_{ij}(\alpha)=\alpha(\log\alpha-1)$ will be used on certain arcs (with  $C_{ij}(0)=0$ and $C_{ij}(\alpha)=\infty$ if $\alpha<0$.)

 The minimum-cost flow problem with separable convex objective is defined as follows.

\begin{align}
\min~\sum_{ij\in E}C_{ij}(f_{ij})&\nonumber\\
\sum_{j:ji\in E}f_{ji}-\sum_{j:ij\in E}f_{ij}&=b_i\quad\forall i\in
V\tag{P}\label{probl}\\
f\ge 0&\nonumber
\end{align}
The problem is often defined with more general lower and upper
capacities: $\ell_{ij}\le f_{ij}\le u_{ij}$. One can reduce general
capacities to the above form via the following standard reduction (see
e.g. \cite [Sec 2.4]{amo}).
For each arc $ij\in E$, let us add a new node $k=k_{ij}$, and replace $ij$ by
two arcs $ik$ and $jk$.
Let us set $b_k=u_{ij}-\ell_{ij}$,
$C_{ik}(\alpha)=C_{ij}(\alpha+\ell_{ij})$, $C_{jk}\equiv 0$. Furthermore,
let us increase $b_i$ by $\ell_{ij}$ and decrease $b_j$ by
$u_{ij}$. 
It is easy to see that this gives an equivalent optimization problem, and if the original graph had $n'$ nodes and $m'$ arcs, the transformed instance has $n=n'+m'$ nodes and $m=2m'$ arcs.

Further, we may assume without loss of generality that $G=(V,E)$ is strongly
connected and (\ref{probl}) is always feasible. Indeed, we can add a new node $t$ with edges $vt$, $tv$ for
any $v\in V$, with extremely high (possibly linear) cost functions on
the edges. This guarantees that an optimal solution shall not use such
edges, whenever the problem is feasible. We will also assume $n\le m$.

By a {\sl pseudoflow} we mean a function  $f:E\rightarrow \R$
satisfying the capacity constraints.
For the uncapacitated problem, it simply means $f\ge 0$.
Let
\begin{equation}
\rho_f(i):=\sum_{j:ji\in E}f_{ji}-\sum_{j:ij\in E}f_{ij},\label{def:rho}
\end{equation}
denote the flow balance at node $i$, and let 
\[
Ex(f)=Ex_b(f):=\sum_{i\in V}\max\{\rho_f(i)-b_i,0\}
\]
 denote the
total positive excess. For an arc set $F$, let $\obe F:=\{ji: ij\in F\}$ denote the set
of reverse arcs,  and let $\ole F=F\cup \obe F$. We shall use the
vector norms $||x||_\infty=\max |x_i|$ and $||x||_1=\sum |x_i|$.

Following \cite{Hochbaum90} and \cite{Karzanov97}, we do not 
require the functions $C_{ij}$ to be given explicitly, but assume
oracle access only.

\begin{oracle}\label{assump:oracle}
We are given an oracle, that we will refer to as the {\sl differential
  oracle}, satisfying either of the following
properties.
\begin{enumerate}[(a)]
\item For every arc $ij\in E$, the oracle returns the value $C'_{ij}(\alpha)$ in $O(1)$ time
 for every $\alpha\in \R$.
If $\alpha$ is rational then $C'_{ij}(\alpha)$ is also rational.
\item
For every arc $ij\in E$, the oracle returns the value $e^{C'_{ij}(\alpha)}$ in $O(1)$ time
 for every $\alpha\in \R$.
If $\alpha$ is rational then $e^{C'_{ij}(\alpha)}$ is also rational.
\end{enumerate}
\end{oracle}

These two options are tailored to main the applications. The more natural  Oracle~\ref{assump:oracle}(a) holds for 
 quadratic objectives, where $C'_{ij}(\alpha)=2c_{ij}\alpha+d_{ij}$ for the cost function $C_{ij}(\alpha)=c_{ij}\alpha^2+d_{ij}$.
 Option (b) is needed  for  Fisher markets, where
 $C'_{ij}(\alpha)=\log \alpha$  for cost functions of the form $C_{ij}(\alpha)=\alpha(\log \alpha-1)$; and
 $C'_{ij}(\alpha)=-\log U_{ij}$ for the other type of cost function,
 $C_{ij}(\alpha)=-\alpha\log U_{ij}$, for a rational $U_{ij}$.
 Note that we do not assume an evaluation oracle returning
$C_{ij}(\alpha)$ or $e^{C_{ij}(\alpha)}$ - these values are not needed for the
algorithm.

The next assumption slightly restricts the class of functions $C_{ij}$
for technical reasons.
\begin{align}
\begin{tabular}{l}
\mbox{Each cost function $C_{ij}(\alpha)$ is either linear or strictly convex, that is,}\\
\mbox{$C'_{ij}(\alpha)$ is either constant or strictly monotone increasing}.
\end{tabular}\tag{$\star$}\label{assump:lin-nonlin-igazi}
\end{align}

Arcs  with $C_{ij}(\alpha)$ linear are called {\sl linear arcs}, the rest is
called {\sl nonlinear arcs}. Let $m_L$ and $m_N$ denote their numbers,
respectively. We use the terms linear and nonlinear for the
corresponding reverse arcs
as well. If $C_{ij}$ does not satisfy this restriction, $\R$ can be decomposed into intervals
such that $C'_{ij}$ is either constant or strictly monotone increasing
on each interval. We can replace $ij$ by a set of paths of length two
(to avoid adding parallel arcs) with appropriately chosen capacities
and cost functions all of which satisfy the assumption. Indeed, the
piecewise linear utility functions in Fisher markets with spending
constraint utilities will be handled in a similar way. 
If the cost functions are explicitly given, for example, the slope of
every linear segment is part of the input, then the size of the
resulting network still only depends on the input size (that includes
all numbers in the input). Hence a strongly polynomial algorithm on
this instance will be strongly polynomial with respect to the original
instance as well. This does not hold however if the functions $C_{ij}$
is given in some different, implicit way.


\subsection{Optimality and $\Delta$-feasibility}\label{sec:opt}
Given a pseudoflow $f$, let us define the residual graph $E_f$ as
\[E_f:=E\cup\{ij: ji\in E, f_{ij}>0\}.\]
Arcs in $E$ are called {\sl forward arcs}, and those in the second set
{\sl backward arcs}. 
Recall our assumption that the graph contains no pairs of oppositely
directed arcs, hence the backward arcs are not contained in $E$.
We will use the convention that on a backward arc $ji$,
$f_{ji}=-f_{ij}$, and
$C_{ji}(\alpha)=C_{ij}(-\alpha)$, also convex and differentiable. The {\sl residual
capacity} is $\infty$ on forward arcs and $f_{ij}$ on the backward arc
$ji$.

The Karush-Kuhn-Tucker conditions assert that the feasible solution $f$ to
(\ref{probl}) is optimal if and only if there exists a potential
vector $\pi:V\rightarrow \R$ such that 
\begin{equation}\label{cond:opt}
\pi_j-\pi_i\le C_{ij}'(f_{ij})\ \ \ \ \forall ij\in E_f.
\end{equation}
This is equivalent to asserting that the residual graph contains no
negative cost directed cycles with respect to the cost function $C_{ij}'(f_{ij})$.

For a value $\Delta>0$, let 
\[
E_f(\Delta)=E\cup\{ij: ji\in E, f_{ij}\ge\Delta\}
\]
 denote the subset of arcs in $E_f$  that
have residual capacity at least $\Delta$.
We say that the pseudoflow $f$ is {\em $\Delta$-feasible}, if there exists a
potential vector
$\pi:V\rightarrow \R$ such that 
\begin{equation}
\pi_j-\pi_i\le C'_{ij}(f_{ij}+\Delta)\ \ \ \ \forall ij\in E_f(\Delta).\label{eq:delta-feas}
\end{equation}
Equivalently, $f$ is $\Delta$-feasible if and only if $E_f(\Delta)$ contains
no negative cycles with respect to the cost function
$C'_{ij}(f_{ij}+\Delta)$. 
If $ji$ is a reverse arc, then (\ref{eq:delta-feas})  gives 
$C'_{ij}(f_{ij}-\Delta)\le \pi_j-\pi_i$.

We note that our notion is different (and weaker) than the analogous conditions in \cite{Minoux86}
and in \cite{Hochbaum90},
where $(C_{ij}(f_{ij}+\Delta)-C_{ij}(f_{ij}))/\Delta$ is used in the place of $C'_{ij}(f_{ij}+\Delta)$.

\medskip

\begin{algorithm}[htb]
\begin{tabbing}
xxxxx \= xxx \= xxx \= xxx \= xxx \= \kill
\> \textbf{Subroutine} \textsc{Adjust}$(\Delta,\bar f)$\\
\> \textbf{INPUT} A $2\Delta$-feasible pseudoflow $\bar f$ and a potential vector $\pi$ satisfying (\ref{eq:delta-feas}) with $\bar f$ and $2\Delta$.\\
\> \textbf{OUTPUT} A $\Delta$-feasible pseudoflow $f$ such that $\pi$ satisfies (\ref{eq:delta-feas}) with $f$ and $\Delta$.\\
\> \textbf{for all} $ij\in E$ \textbf{do} \\
\> \> \textbf{if} $C'_{ij}(\bar f_{ij}+\Delta)<\pi_j-\pi_i$
\textbf{then} $f_{ij}\leftarrow \bar f_{ij}+\Delta$.\\
\> \> \textbf{elseif} $\bar f_{ji}\ge \Delta$ and $\pi_j-\pi_i<C'_{ij}(\bar f_{ij}-\Delta)$
\textbf{then} $f_{ij}\leftarrow \bar f_{ij}-\Delta$.\\
\> \> \> \textbf{else} $f_{ij}\leftarrow \bar f_{ij}$.\\
\> \textbf{return} $f$.
\end{tabbing}
\caption{}
\label{fig:adjust}
\end{algorithm}

The subroutine \textsc{Adjust}($\Delta,f$) (see Algorithm~\ref{fig:adjust})
transforms a 
$2\Delta$-feasible pseudoflow to a $\Delta$-feasible pseudoflow by
possibly changing the value of
every arc by  $\pm \Delta$. 

\begin{lemma}\label{lem:adjust}
The subroutine \textsc{Adjust}($\Delta,f$) is well-defined and correct: it returns a
$\Delta$-feasible pseudoflow with $(f,\pi)$ satisfying (\ref{eq:delta-feas}). Further, $Ex(f)\le Ex(\bar f)+m_N\Delta$
(recall that $m_N$ is the number of nonlinear arcs).
\end{lemma}
\begin{proof}
First we observe that the ``if'' and ``elseif'' conditions cannot hold
simultaneously: $C'_{ij}(\bar f_{ij}+\Delta)<\pi_j-\pi<C'_{ij}(\bar
f_{ij}-\Delta)$ would contradict the convexity of $C_{ij}$. 
Consider the potential vector $\pi$ satisfying (\ref{eq:delta-feas}) with $\bar
f$ and $2\Delta$. We prove that $\pi$ satisfies
(\ref{eq:delta-feas}) with $f$ and $\Delta$ as well.

First, take a forward arc  $ij\in E$ with
$C'_{ij}(\bar f_{ij}+\Delta)<\pi_j-\pi_i$. By
$2\Delta$-feasibility we know $\pi_j-\pi_i\le C'_{ij}(\bar
f_{ij}+2\Delta)$.
These show that setting $f_{ij}=\bar f_{ij}+\Delta$
 satisfies (\ref{eq:delta-feas}) for both $ij$ and $ji$, using
\[C'_{ij}(f_{ij}-\Delta)\le C'_{ij}(f_{ij})=C'_{ij}(\bar f_{ij}+\Delta)<\pi_j-\pi_i\le C'_{ij}(\bar f_{ij}+2\Delta)=C'_{ij}(f_{ij}+\Delta).\]

Next, assume  $\bar f_{ji}\ge \Delta$ and $\pi_j-\pi_i<C'_{ij}(\bar
f_{ij}-\Delta)$.
Note that $f_{ij}$
satisfies (\ref{eq:delta-feas}) by $\pi_j-\pi_i<C'_{ij}(\bar
f_{ij}-\Delta)\le C'_{ij}(\bar f_{ij})=C'_{ij}(f_{ij}+\Delta)$.

If $ji\in E_{\bar f} (2\Delta)$ (that is, $\bar f_{ij}\ge 2\Delta$), then we have 
$C'_{ij}(f_{ij}-\Delta)=C'_{ij}(\bar f_{ij}-2\Delta)\le \pi_j-\pi_i$,
and thus (\ref{eq:delta-feas}) also holds for  $ji$.
If $ji\in E_{\bar f} (\Delta)-E_{\bar f} (2\Delta)$, then
$ji\notin E_f(\Delta)$.

Finally, consider the case when $f_{ij}=\bar f_{ij}$. The condition
(\ref{eq:delta-feas}) holds for $ij$ as we assume
 $\pi_j-\pi_i\le C'_{ij}(\bar f_{ij}+\Delta)$. Also, either
 $f_{ij}=\bar f_{ij}<\Delta$ and thus $ji\notin E_f(\Delta)$, or
$f_{ij}=\bar f_{ij}\ge \Delta$ and  (\ref{eq:delta-feas}) holds for
$ji$ by the assumption  $C'_{ij}(\bar
f_{ij}-\Delta)\le \pi_j-\pi_i$.

To verify the last claim, observe that $C'_{ij}$ is constant on every
linear arc and therefore $\bar f_{ij}=f_{ij}$ will be set on every
linear arc. The flow change is $\pm\Delta$ on every nonlinear arc;
every such change may increase the excess of one of the endpoints of
the arc by $\Delta$. Consequently,  $Ex(f)\le Ex(\bar f)+m_N\Delta$ follows.
\end{proof}

\section{The  basic algorithm}\label{sec:basic}
Algorithm~\ref{fig:basic} outlines a simple algorithm for minimum cost
flows with separable convex objectives, to be referred as the
``Basic algorithm''. This is a modified version of Minoux's algorithm \cite{Minoux86}.
The algorithm returns a $\varepsilon$-accurate solution for a required precision $\varepsilon>0$. That is, for output $f$, there is an optimal solution $f^*$ such that $\|f-f^*\|_\infty<\varepsilon$.

\begin{algorithm}[htb]

\begin{tabbing}
xxxxx \= xxx \= xxx \= xxx \= xxx \= \kill
\> \textbf{Algorithm} \textsc{Basic}\\
\> $f\leftarrow {\bf 0}$; $\Delta\leftarrow \Delta_0$;\\
\> \textbf{do}  \ \ \ {\sl\small  \slash \slash $\Delta$-phase}\\
\> \> \textbf{do} \ \ \ {\sl\small  \slash \slash main part}\\
\> \> \> $S(\Delta)\leftarrow\{i\in V: \rho_f(i)-b_i\ge \Delta\}$;\\
\> \> \> $T(\Delta)\leftarrow\{i\in V: \rho_f(i)-b_i\le -\Delta\}$;\\
\> \> \> $P\leftarrow$ shortest $s-t$ path in $E_f(\Delta)$ for the
cost $C'_{ij}(f_{ij}+\Delta)$ with $s\in S(\Delta)$,
$t\in T(\Delta)$;\\
\> \> \> send $\Delta$ units of flow on $P$ from $s$ to $t$;\\
\> \>  \textbf{while} $S(\Delta),T(\Delta)\neq \emptyset$;\\
\> \> \textsc{Adjust}$(\Delta/2,f)$;\\
\> \>  $\Delta\leftarrow\Delta/2$;\\
\> \textbf{while} $\Delta>\varepsilon/(2n+m_N+1)$;\\
\> \textbf{Return} $f$.
\end{tabbing}
\caption{}\label{fig:basic}
\end{algorithm}

We start with the pseudoflow $f\equiv {\bf 0}$ and an initial value $\Delta=\Delta_0$.
We assume that the value $\Delta_0$ is provided in the input so that ${\bf 0}$ is a $\Delta_0$-feasible and $Ex({\bf 0})\le (2n+m)\Delta_0$;
in the enhanced algorithm we shall specify how such a  $\Delta_0$ value can be determined.
The algorithm consists of $\Delta$-phases, with $\Delta$ decreasing by  a
factor of two between two phases; the algorithm terminates once $\Delta<\varepsilon/(2n+m_N+1)$.

In the main part of phase $\Delta$, 
let $S(\Delta)=\{i\in V: \rho_f(i)-b_i\ge \Delta\}$ and
$T(\Delta)=\{i\in V: \rho_f(i)-b_i\le -\Delta\}$, the set of nodes
with excess and deficiency at least $\Delta$. As long as $S(\Delta)\neq \emptyset$,
$T(\Delta)\neq\emptyset$, send $\Delta$ units of flow from a node
$s\in S(\Delta)$ to a node $t\in T(\Delta)$ on a shortest path in
$E_f(\Delta)$ with respect to the cost function $C'_{ij}(f_{ij}+\Delta)$.
(Note that there must be a path connecting nodes in $S(\Delta)$ and
$T(\Delta)$, due to our assumption that the graph $G=(V,E)$ is
strongly connected, and $E\subseteq E_f(\Delta)$.)

The main part finishes once $S(\Delta)=\emptyset$ or
$T(\Delta)=\emptyset$. The $\Delta$-phase terminates by 
performing  \textsc{Adjust}$(\Delta/2,f)$ and proceeding to the
next phase with scaling factor $\Delta/2$.

In the main part, we need to compute shortest paths in the graph
$E_f(\Delta)$ for the cost function $C'_{ij}(f_{ij}+\Delta)$. This can
be done only if there is no negative cost cycle.  $\Delta$-feasibility
is exactly this property and is maintained throughout (see
Lemma~\ref{lem:basic-alg} below). Details of the shortest path
computation will be given in Section~\ref{sec:shortest}(ii),
for the enhanced algorithm.


\subsection{Analysis}
\begin{theorem}\label{thm:basic-running}
The Basic algorithm delivers an $\varepsilon$-accurate solution in $O(\log((2n+m_N+1)\Delta_0/\varepsilon)$ phases, and every phase comprises at most $O(2n+m_N)$ flow augmentations.
\end{theorem}
An appropriate $\Delta_0$ can be chosen to be polynomial in the input size, hence this gives a weakly polynomial running time bound.
We now state the two simple lemmas needed to prove this theorem.
The first lemma verifies the correctness and efficiency of the algorithm, showing that $\Delta$-feasibility is maintained throughout
and the number of flow augmentations is linear in every $\Delta$-phase.
We omit the proof; its analogous counterpart for the enhanced algorithm will be proved in Lemma~\ref{lem:enhanced}.
\begin{lemma}\label{lem:basic-alg}
\begin{enumerate}[(i)]
\item
In the main part of the $\Delta$-phase, the pseudoflow is an integer
  multiple of $\Delta$ on each arc, and consequently,
  $E_f(\Delta)=E_f$.
\item $\Delta$-feasibility is maintained when augmenting on a shortest
  path. 
\item 
 At the beginning of the main part, $Ex(f)\le (2n+m_N)\Delta$, and at
 the end, $Ex(f)\le n\Delta$. 
\item The main part consists of at most
 $2n+m_N$ flow augmentation steps.
\end{enumerate}
\end{lemma}
Our second lemma asserts the proximity of a current flow to all later flows during the algorithm. If we let the algorithm run without ever terminating, it will converge to an optimal solution. Hence the lemma justifies that the algorithm obtains an $\varepsilon$-accurate solution as claimed in Theorem~\ref{thm:basic-running}. Moreover, it also helps to identify edges which
must be contained in the support of an optimal solution. The proof is also omitted; see Lemma~\ref{lem:F-in-F} and the first part of the proof of Theorem~\ref{thm:main-bound}. This is essentially the same argument that was used by Orlin (e.g. \cite[Lemma 10.21]{amo}). 

\begin{lemma}\label{lem:f-change}
Let $f$ be the pseudoflow at the end of the main part of the $\Delta$-phase and $f'$ in an
arbitrary later phase. Then 
$||f-f'||_\infty\le (2n+m+1)\Delta$. If $f_{ij}> (2n+m+1)\Delta$ at
the end of the $\Delta$-phase, then this property is maintained in all
later phases, and there exists an optimal solution
$f^*$ with $f^*_{ij}>0$.
\end{lemma}
For all such arcs, we can conclude $\pi_j-\pi_i=C'_{ij}(f^*_{ij})$ for an optimal solution $f^*$. 
It will belong to the set of revealed arcs, defined in the next section. The overall aim of the algorithm is 
to identify a large enough set of revealed arcs containing the support of an optimal solution. The above lemma guarantees that the first such arc can be identified in a strongly polynomial number of steps in the Basic algorithm. We will however need to modify the algorithm in order to guarantee that the set of revealed arcs is always extended in a strongly polynomial number of steps.

\section{The enhanced algorithm}\label{sec:enhanced}
\subsection{Revealed arc sets}\label{sec:revealed}
Let $F^*$ denote the set of arcs  that are tight in every
optimal solution (note that in general, we do not assume the
uniqueness of the optimal solution). This arc set plays a key role in
our algorithm.
Formally, 
\begin{equation}
\begin{aligned}
F^*:=\{ij\in E: \ \pi_j-\pi_i=C'_{ij}(f_{ij})\mbox{ holds }&\forall f \mbox{ optimal to (\ref{probl}), }\forall \pi:V\to\R,\\
&\mbox{ s.t. } (f,\pi) \mbox{ satisfies the inequalities (\ref{cond:opt})}\}. 
\end{aligned} \label{def:Fcs}
\end{equation}
The next lemma
shows that $F^*$ contains the support of every optimal solution.

\begin{lemma}\label{lem:canon-F}
Let $f$ be an arbitrary optimal solution to (\ref{probl}), and
$f_{ij}>0$ for some $ij\in E$.
Then $ij\in F^*$.
\end{lemma}
The proof needs the following notion, also used later.
Let $x,y:E\rightarrow \R$ be two vectors.
Let us define the {\sl  difference graph} $D_{x,y}=(V,E_{x,y})$ with
$ij\in E_{x,y}$ if $ij\in E$ and $x_{ij}>y_{ij}$ or if $ji\in E$ and $x_{ji}<y_{ji}$. Using the
convention $x_{ji}=-x_{ij}$, $y_{ji}=-y_{ij}$ it follows that
$x_{ij}>y_{ij}$ for every $ij\in E_{x,y}$.
We will need the following simple claim.
\begin{claim}\label{cl:cycle}
Assume that for two vectors $x,y:E\rightarrow \R$, $\rho_x=\rho_y$
holds (recall the definition of $\rho$ in (\ref{def:rho})). Then every
arc in the difference graph $E_{x,y}$  must be contained in a cycle in $E_{x,y}$.
\end{claim}
\begin{proof}
For $ij\in E_{x,y}$, let us set $z_{ij}=x_{ij}-y_{ij}$ if $x_{ij}>y_{ij}$. The assumption
$\rho_x=\rho_y$ implies that $z_{ij}$ is a circulation in $E_{x,y}$
with positive value on every arc. As such, it can be written as a
nonnegative combination of incidence vectors of cycles.
Therefore every $ij\in E_{x,y}$ must be contained in a cycle.
\end{proof}

\begin{proof}[Proof of Lemma~\ref{lem:canon-F}]
Let $f^*$ be another arbitrary optimal solution, and consider potentials
$\pi$ and $\pi^*$ with both $(f,\pi)$ and $(f^*,\pi^*)$ satisfying (\ref{cond:opt}).
We shall prove that $\pi^*_j-\pi^*_i=C'_{ij}(f^*_{ij})$. Since $(f^*,\pi^*)$ is chosen arbitrarily, this will imply $ij\in F^*$.
If $f^*_{ij}>0$, then
$ji\in E_{f^*}$ and thus $\pi^*_j-\pi^*_i=C'_{ij}(f^*_{ij})$
must hold.

Assume now $f^*_{ij}=0$.
Consider the  difference graph $D_{f,f^*}$.  Since $f_{ij}>f^*_{ij}$,
it follows that $ij\in E_{f,f^*}$. Because of $\rho_{f^*}\equiv
\rho_{f}\equiv b$, Claim~\ref{cl:cycle} is applicable and provides a cycle $C$
in $E_{f,f^*}$ containing $ij$.
For every arc $ab\in C$, $f_{ab}>f^*_{ab}$ and thus $ab\in E_{f^*}$ and $ba\in E_{f}$.
By (\ref{cond:opt}), 
\begin{align*}
0&=\sum_{ab\in C}\pi^*_b-\pi^*_a\le \sum_{ab\in C}
C'_{ab}(f^*_{ab})\quad \mbox{ and}\\
0&=\sum_{ab\in C}\pi_a-\pi_b\le \sum_{ab\in C} C'_{ba}(f_{ba})=-\sum_{ab\in C}C'_{ab}(f_{ab}).
\end{align*}
The convexity of $C_{ab}$ and $f_{ab}>f^*_{ab}$ give
$C'_{ab}(f_{ab})\ge C'_{ab}(f^*_{ab})$.
In the above inequalities,
 equality must hold everywhere, implying
$\pi^*_j-\pi^*_i=C'_{ij}(f^*_{ij})$ as desired.
\end{proof}

We shall see that using Oracle~\ref{assump:trial} (to be described
later), finding the set $F^*$ 
enables us to compute an optimal solution in strongly polynomial time.
In the Basic algorithm,
$F=\{ij\in E: f_{ij}>(2n+m+1)\Delta\}$ is always a subset of $F^*$
according to  Lemmas~\ref{lem:f-change} and \ref{lem:canon-F}.
Furthermore, once an edge enters $F$, it stays there in
all later phases. 
The Enhanced algorithm provides a modification of the basic algorithm with the guarantee that within
every $O(\log n)$  phases, a new arc enters
$F$.

\medskip

In each step of the enhanced algorithm, there will be an arc set $F$, called the
{\sl revealed arc set}, which is
guaranteed to be a subset of $F^*$. We remove the lower capacity 0
from arcs in $F$ and allow also negative values here. 

Formally, for an edge set $F\subseteq E$, a vector 
$f:E\rightarrow\R$ is an {\sl $F$-pseudoflow}, if $f_{ij}\ge 0$ for
$ij\in E\setminus F$ (but it is allowed to be negative on $F$).
For such an $f$, let us define 
\begin{equation}
E^F_f:=E_f\cup \obe F=E\cup \obe F\cup \{ji: ij\in E\setminus F, f_{ij}>0\}.\label{eq:EfF}
\end{equation}
 If $ij\in F$, then
 the residual capacity of $ji$ is $\infty$.
In every phase of the algorithm, we maintain an $F$-pseudoflow $f$ for
a revealed arc set $F\subseteq F^*$.

Provided the revealed arc set $F\subseteq F^*$, we will aim for $F$-optimal solutions as defined
below; we prove that finding an $F$-optimal solution is essentially
equivalent to finding an optimal one.
We say that $f:E\rightarrow \R$ is {\sl $F$-optimal}, if 
it is an $F$-pseudoflow with $\rho_f\equiv b$ and there exists a potential
vector
$\pi:V\rightarrow \R$ with
\begin{equation}\label{cond:F-opt}
\pi_j-\pi_i\le C_{ij}'(f_{ij})\ \ \ \ \forall ij\in E^F_f.
\end{equation}
This is stronger than the optimality condition (\ref{cond:opt}) in
that it also requires the inequality on arcs in $\obe F$. On the other hand,
it does not imply optimality as it allows $f_{ij}<0$ for $ij\in F$.
Nevertheless, it is easy to see that every optimal solution $f^*$ is
also $F$-optimal for every $F\subseteq F^*$. This is due to the
definition of 
$F^*$ as the set of arcs satisfying $\pi_j-\pi_i=
C_{ij}'(f_{ij})$ whenever $(f,\pi)$ satisfies (\ref{cond:opt}). 
Conversely, we shall prove that provided an $F$-optimal solution, we can easily
find an optimal solution by a single feasible circulation algorithm, a
problem equivalent to maximum flows (see
\cite[Chapters 6.2, 7]{amo}).

\begin{lemma}\label{lem:compute-opt}
Assume that for a subset $F\subseteq F^*$, an $F$-optimal solution
$f$ is provided.
Then an optimal solution to (\ref{probl}) can be found by a feasible
circulation algorithm. Further, $ij\in F^*$ whenever $f_{ij}>0$. 
\end{lemma}
\begin{proof}
Assume $f$ and $\bar f$ are both $F$-optimal solutions, that is, for some vectors $\pi$ and $\bar \pi$, 
the pairs $(f,\pi)$ and $(\bar f,\bar \pi)$ both satisfy (\ref{cond:F-opt}).
We prove that {\em (i)} $f_{ij}=\bar f_{ij}$ whenever $ij$ is a nonlinear
arc; and {\em (ii)} if $ij$ is a linear arc with $f_{ij}\neq \bar
f_{ij}$, then $\pi_j-\pi_i=C'_{ij}(f_{ij})=C'_{ij}(\bar
f_{ij})=\bar\pi_j-\bar\pi_i$. 

Note that {\em (i)} and {\em (ii)} immediately imply the
second half of the claim as it can be applied for $f$ and an arbitrary
optimal (and consequently, $F$-optimal) solution $\bar f$.

The proof uses  the same argument as for
Lemma~\ref{lem:canon-F}. 
W.l.o.g. assume $f_{ij}>\bar f_{ij}$ for an arc $ij$, and
consider the  difference graph $D_{f,\bar f}$.   Since
$\rho_f\equiv\rho_{\bar f}\equiv b$ and $f_{ij}>\bar
f_{ij}$, Claim~\ref{cl:cycle} is applicable and shows that $ij$ must be contained on a cycle $C\subseteq E_{f,\bar
  f}$. For every arc $ab\in C$, $ab\in E^F_{\bar f}$ and $ba\in E^F_f$
follows (using $\ole F\subseteq   E^F_{\bar f}\cap E^F_f$).
By (\ref{cond:F-opt}), 
\begin{align*}
0&=\sum_{ab\in C}\bar\pi_b-\bar\pi_a \le \sum_{ab\in C} C'_{ab}(\bar f_{ab})\quad\mbox{ and}\\
0&= \sum_{ab\in C}\pi_a-\pi_b\le \sum_{ab\in C} C'_{ba}(f_{ba})=-\sum_{ab\in C}C'_{ab}(f_{ab}).
\end{align*}
Now convexity yields 
$C'_{ab}(f_{ab})=C'_{ab}(\bar f_{ab})$ for all $ab\in C$.
The condition (\ref{assump:lin-nonlin-igazi}) implies that all
arcs in $C$ are linear, in particular, $ij$ is linear. This immediately proves
{\em (i)}. To verify {\em (ii)}, observe that all above inequalities must hold with
equality.

This suggests the following simple method to transform an $F$-optimal
solution $f$ to an optimal $f^*$ of (\ref{probl}). For every nonlinear arc $ij$, we
must have $f^*_{ij}=f_{ij}$. Let $H\subseteq E$ be the set of linear
arcs satisfying $\pi_j-\pi_i=C'_{ij}(f_{ij})$. 
Consider the solutions $h$ of the following feasible circulation problem:
\begin{align*}
h_{ij}&=f_{ij}\quad \forall ij\in E\setminus H\\
\sum_{j:ji\in E}h_{ji}-\sum_{j:ij\in E}h_{ij}&=b_i\quad\forall i\in V\\
h&\ge 0
\end{align*}
%
We claim that the feasible solutions to this circulation problem are
precisely the optimal solutions to (\ref{probl}).
Indeed, if $f^*$ is an optimal solution, then  {\em (i)} 
and {\em (ii)} imply that $f^*_{ij}=f_{ij}$ for all $ij\in E\setminus
H$ and $ij\in H$ for every arc with $f_{ij}\neq f^*_{ij}$. The degree
conditions are satisfied because of $\rho_{f^*}\equiv\rho_f\equiv b$.
Conversely, every feasible
circulation $h$ is an optimal solution to (\ref{probl}), since $(h,\pi)$ satisfies
(\ref{cond:opt}).
\end{proof}

In every step of our algorithm we will have a scaling parameter
$\Delta\ge 0$ and a revealed arc set $F\subseteq F^*$. The Basic algorithm
used the
notion of $\Delta$-feasibility; it has
to be modified according to $F$. 
Let $E_f^F(\Delta)$ denote
the set of arcs in $E_f^F$ with residual capacity at least
$\Delta$. That is,
\begin{equation}\label{def:F-delta}
E_f^F(\Delta):=E_f(\Delta)\cup \obe F=E\cup \obe F\cup \{ji: ij\in E\setminus F, f_{ij}\ge \Delta\}.
\end{equation}
We say that the
$F$-pseudoflow $f$ is {\sl $(\Delta,F)$-feasible},  
if there exists a potential vector
$\pi:V\rightarrow \R$ so that 
\begin{equation}
\pi_j-\pi_i\le C'_{ij}(f_{ij}+\Delta)\ \ \ \ \forall ij\in E^F_f(\Delta).\label{eq:delta-F-feas}
\end{equation} 
This is equivalent to the property that $E_f^F(\Delta)$ contains
no negative cycle with respect to the cost function
$C'_{ij}(f_{ij}+\Delta)$. 

In accordance with $(\Delta,F)$-feasibility, we have to modify the subroutine \textsc{Adjust}. The modified subroutine, denoted by
\textsc{Adjust}$(\Delta,f,F)$ is shown in Algorithm~\ref{fig:adjust-F}. The only difference from  Algorithm~\ref{fig:adjust}
is that the condition (\ref{eq:delta-feas}) is replaced by (\ref{eq:delta-F-feas}), and that in the second condition, ``$\bar f_{ji}\ge \Delta$" is replaced by ``$\bar f_{ji}\ge \Delta$ or $ij\in F$".
The following lemma can be proved by the same argument as Lemma~\ref{lem:adjust}.
\begin{algorithm}[htb]
\begin{tabbing}
xxxxx \= xxx \= xxx \= xxx \= xxx \= \kill
\> \textbf{Subroutine} \textsc{Adjust}$(\Delta,\bar f,F)$\\
\> \textbf{INPUT} A $(2\Delta,F)$-feasible pseudoflow $\bar f$ and a potential vector $\pi$ satisfying (\ref{eq:delta-F-feas}) with $\bar f$ and $2\Delta$.\\
\> \textbf{OUTPUT} A $(\Delta,F)$-feasible pseudoflow $f$ such that $\pi$ satisfies (\ref{eq:delta-F-feas}) with $f$ and $\Delta$.\\
\> \textbf{for all} $ij\in E$ \textbf{do} \\
\> \> \textbf{if} $C'_{ij}(\bar f_{ij}+\Delta)<\pi_j-\pi_i$
\textbf{then} $f_{ij}\leftarrow \bar f_{ij}+\Delta$.\\
\> \> \textbf{elseif}  ($\bar f_{ji}\ge \Delta$ or $ij\in F$) and $\pi_j-\pi_i<C'_{ij}(\bar f_{ij}-\Delta)$
\textbf{then} $f_{ij}\leftarrow \bar f_{ij}-\Delta$.\\
\> \> \> \textbf{else} $f_{ij}\leftarrow \bar f_{ij}$.\\
\> \textbf{return} $f$.
\end{tabbing}
\caption{}\label{fig:adjust-F}
\end{algorithm}

\begin{lemma}\label{lem:adjust-F}
The subroutine \textsc{Adjust}$(\Delta,f,F)$ is well-defined and correct: it returns a
$(\Delta,F)$-feasible pseudoflow with $(f,\pi)$ satisfying (\ref{eq:delta-F-feas}). Further, $Ex(f)\le Ex(\bar f)+m_N\Delta$.
\end{lemma}

\medskip

Finally, we say that a set $F\subseteq E$ is {\sl linear acyclic}, if 
$F$ does not contain any undirected cycles of linear arcs (that is, no
cycle in $F$ may consist of linear arcs and their reverse arcs). We
shall maintain that the set of revealed arcs, $F$ is linear acyclic.

This notion is motivated by the following: assume there exists a cycle consisting of linear arcs and their reverses.
Given an $F$-pseudoflow, we could modify it by sending an arbitrary amount of flow around this cycle. Hence we would not be able to derive our proximity result Lemma~\ref{lem:f1f2} and  Lemma~\ref{lem:error-bound} that relies on it. On the other hand, we can pick an arbitrary arc on a cycle of linear arcs, remove it from $F$, an reroute its entire flow on the rest of the cycle.

\subsection{Subroutine assumptions}\label{sec:assump}

Given the set  $F\subseteq F^*$ of revealed arcs, we will try to find
out whether $F$ already contains the support of an optimal solution. This motivates the following definition.
 We say that the (not necessarily nonnegative) vector $x:E\rightarrow
\R$ is {\sl $F$-tight}, 
if $x_{ij}=0$ whenever $ij\notin F$ and there
exists a potential vector $\pi:V\rightarrow \R$ with
\begin{equation}
\pi_j-\pi_i=C'_{ij}(x_{ij}) \quad\forall ij\in F. \label{eq:tight}
\end{equation}
For example, any optimal solution is $F^*$-tight by Lemma~\ref{lem:canon-F}.
Notice that an $F$-tight vector $f$ is not necessarily  $F$-optimal as (\ref{cond:F-opt})
 might be violated for edges in $E_f^F\setminus \ole F$ and also since
$Ex_b(f)>0$ is  allowed (note that $\rho_f\equiv b$ is equivalent to $Ex_b(f)=0$). Conversely, an
$F$-optimal vector is not necessarily $F$-tight as it can be nonzero on
$E\setminus F$.

Given $F$ and some node demands $\hat b:V\rightarrow \R$, we would like to find an $F$-tight $x$ with
$Ex_{\hat b}(x)=0$. This is equivalent to finding a feasible solution $(x,\pi)$
to the following system:
\begin{align}
\pi_j-\pi_i&=C'_{ij}(x_{ij}) \quad\forall ij\in F \notag\\
\sum_{j:ji\in F}x_{ji}-\sum_{j:ij\in E}x_{ij}&=\hat b_i\quad\forall i\in V\label{sys:trial}\\
x_{ij}&=0\quad \forall ij\in E\setminus F\notag
\end{align}
Let us define the {\sl discrepancy} $D_{\hat b}(F)$ of $F$ as the maximum of
$|\sum_{i\in K}\hat b_i|$ over undirected connected components $K$ of $F$.
A trivial necessary condition for solvability is $D_{\hat b}(F)=0$:
indeed, summing up the second set of equalities for a component $K$, we
obtain $0=\sum_{i\in K}\hat b_i$.

\begin{oracle} \label{assump:trial}
Assume we have a subroutine \textsc{Trial}$(F,\hat b)$ so that for any
linear acyclic $F\subseteq E$ and any vector $\hat b:V\rightarrow \R$ satisfying
$D_{\hat b}(F)=0$,
  it  delivers an $F$-tight solution $x$ to (\ref{sys:trial}) with
$\rho_x\equiv {\hat b}$  in strongly polynomial running time $\rho_T(n,m)$.
\end{oracle}
For quadratic cost functions and also for Fisher markets, this
subroutine can be implemented by solving simple systems of equations
(for quadratic, this was already  outlined in Section~\ref{sec:outline-strong}).

\medskip

Consider now an $F$-tight vector $f$, and let
\begin{equation}
err_F(f):=\inf\{\Delta: f\mbox{ is  }(\Delta,F)\mbox{-feasible}\}.\label{def:err}
\end{equation}
Recall the definition (\ref{def:F-delta}) of the edge set
$E_f^F(\Delta)$. 
As $f$ is assumed to be $F$-tight and therefore $f_{ij}>0$ only if
$ij\in F$, we get that $E_f^F(\Delta)=E\cup \overleftarrow
F$. Consequently,  $E_f^F(\Delta)$ is independent of  the value of $\Delta$. 
 Because of continuity, this infimum is
  actually a minimum whenever the set is nonempty. If $f$ is not
  $(\Delta,F)$-feasible
for any $\Delta$, then
let $err_F(f)=\infty$.
$f$ is $F$-optimal if and only if $f$
is a feasible flow 
(that is, $Ex_b(f)=0$) and $err_F(f)=0$.


\begin{oracle} \label{assump:error}
Assume  a subroutine \textsc{Error}($f,F$) is provided, that returns
 $err_F(f)$ 
 for any
$F$-tight vector $f$
in strongly polynomial running time $\rho_E(n,m)$.
Further, if 
$err_\emptyset({\bf 0})=\infty$, then (\ref{probl}) is unbounded.
\end{oracle}
This subroutine seems significantly harder to
implement for the applications: we need to solve a minimum cost-to-time ratio cycle problem
for quadratic costs and all pairs shortest paths for the Fisher markets.

\medskip
Having formulated all necessary assumptions, we are finally in the position to formulate the main result of the paper.

\begin{theorem}\label{thm:main}
Assume Oracles 1-3 are provided and (\ref{assump:lin-nonlin-igazi}) holds for the problem (\ref{probl}) in a network on $n$ nodes and $m$ arcs, $m_N$ among them having nonlinear cost functions.
Let $\rho_T(n,m)$ and $\rho_E(n,m)$ denote the running time of  Oracle~\ref{assump:trial} and Oracle~\ref{assump:error}, and let $\rho_S(n,m)$ be the running time needed for a single shortest path
computation for nonnegative arc lengths. Then an exact optimal solution can be found in  
$O((n+m_N)(\rho_T(n,m)+\rho_E(n,m))+(n+m_N)^2\rho_S(n,m) \log m)$
time.

This gives an $O(m^4\log m)$ algorithm for quadratic convex objectives. For Fisher markets, we obtain
$O(n^4+n^2(m+n\log n)\log n)$ running time for linear and $O(mn^3+m^2(m+n\log
n)\log m)$ for spending constraint utilities. 
\end{theorem}

\subsection{Description of the enhanced algorithm}

\begin{algorithm}[htb]
\begin{tabbing}
xxxxx \= xxx \= xxx \= xxx \= xxx \= \kill
\> \textbf{Algorithm} \textsc{Enhanced Convex Flow}\\
\> \textsc{Error}$({\bf 0},\emptyset)$;\\
\> $f\leftarrow {\bf 0}$; $\Delta\leftarrow \max\{err_\emptyset({\bf 0}),Ex_b({\bf 0})/(2n+m_N)\}$; $F\leftarrow \emptyset$;\\
\> \textbf{repeat} \ \ \ {\sl\small  \slash \slash $\Delta$-phase}\\
\> \> \textbf{do} \ \ \ {\sl\small  \slash \slash main part}\\
\> \> \> $S(\Delta)\leftarrow\{i\in V: \rho_f(i)-b_i\ge \Delta\}$;\\
\> \> \> $T(\Delta)\leftarrow\{i\in V: \rho_f(i)-b_i\le -\Delta\}$;\\
\> \> \> $P\leftarrow$ shortest $s-t$ path in $E^F_f(\Delta)$ for the
cost $C'_{ij}(f_{ij}+\Delta)$ with $s\in S(\Delta)$,
$t\in T(\Delta)$;\\
\> \> \> send $\Delta$ units of flow on $P$ from $s$ to $t$;\\
\> \>  \textbf{while} $S(\Delta),T(\Delta)\neq \emptyset$;\\
\> \> \textsc{Extend}$(\Delta,f,F)$; \\
\> \> \textbf{if} ($F$ was extended) \textsl{and} ($D_b(F)\le \Delta$) \textbf{then} \textsc{Trial-and-Error}$(F)$\\
\> \> \> \> \textbf{else} \textsc{Adjust}$(\Delta/2,f,F)$;\\
\> \> \> \> \> $\Delta\leftarrow\Delta/2$;
\end{tabbing}

\begin{tabbing}
xxxxx \= xxx \= xxx \= xxx \= xxx \= \kill
\> \textbf{Subroutine} \textsc{Extend}$(\Delta,f,F)$\\
\> \textbf{for all} $ij\in E\setminus F$, $f_{ij}>(2n+m+1)\Delta$ \textbf{do} \\
\> \> \textbf{if} $F\cup\{ij\}$ is linear acyclic \textbf{then} $F\leftarrow F\cup\{ij\}$\\
\> \> \textbf{else}\\
\> \> \> $P\leftarrow$ path of linear arcs in $\ole F$ between $i$ and $j$;\\
\> \> \> send $f_{ij}$ units of flow on $P$ from $i$ to $j$;\\
\> \> \> $f_{ij}\leftarrow 0$;
\end{tabbing}
\caption{}\label{alg:enhanced}
\end{algorithm}

Algorithm~\ref{alg:enhanced} starts with  $f={\bf 0}$, $\Delta=\max\{err_\emptyset({\bf 0}), Ex_b({\bf 0})/(2n+m_N)\}$ and $F=\emptyset$.
The algorithm
consists of $\Delta$-phases. In the
$\Delta$-phase, we shall maintain a linear acyclic revealed arc set
$F\subseteq F^*$, and a $(\Delta,F)$-feasible
$F$-pseudoflow $f$. The algorithm will always terminate during the subroutine \textsc{Trial-and-Error}.

The main part of the $\Delta$-phase  is the same as in the Basic
algorithm. 
Let $S(\Delta)=\{i\in V: \rho_f(i)-b_i\ge \Delta\}$ and $T(\Delta)=\{i\in V: \rho_f(i)-b_i\le -\Delta\}$. As long as $S(\Delta)\neq \emptyset$,
$T(\Delta)\neq\emptyset$, send $\Delta$ units of flow from a node
$s\in S(\Delta)$ to a node $t\in T(\Delta)$ on a shortest path in $E_f^F(\Delta)$ with respect to the cost function $C'_{ij}(f_{ij}+\Delta)$.
(The existence of such a path $P$ is guaranteed by our assumption that the graph $G=(V,E)$ is
strongly connected.)

After the main part (the sequence of path augmentations) is finished, the subroutine \textsc{Extend}$(\Delta,f,F)$ 
adds new arcs 
$ij\in E\setminus F$ with 
$f_{ij}>(2n+m+1)\Delta$ to $F$ maintaining the linear acyclic property. This is achieved as follows: we first add all nonlinear such arcs to $F$.
We add a linear arc to $F$ if it does not create any (undirected) cycles in $F$. If adding the linear arc $ij$ would create a cycle, we do not include it in $F$, but reroute the entire flow from $ij$ using the (undirected) path in $F$ between $i$ and $j$.

If no new arc enters $F$, then we perform \textsc{Adjust}$(\Delta/2,f,F)$ and move to the next scaling phase with
the same $f$ and set the scaling factor to $\Delta/2$. This is
done also if $F$ is extended, but it has a high discrepancy: $D_b(F)>\Delta$.

Otherwise, the subroutine \textsc{Trial-and-Error}$(F)$
determines the next $f$ and $\Delta$. Based on the arc set
$F$, we find a new $F$-pseudoflow  $f$ and
scaling factor at most $\Delta/2$. The subroutine may also terminate
with an $F$-optimal solution, which enables us to find an optimal
solution to (\ref{probl}) by a maximum flow computation due to Lemma~\ref{lem:compute-opt}.
Theorem~\ref{thm:main-bound} will show that this is guaranteed to
happen within a strongly polynomial number of steps.

\subsubsection*{The Trial-and-Error subroutine}
The subroutine assumes that 
 the discrepancy of $F$ is small:
$D_b(F)\le \Delta$.

{\bf Step 1.} First,  modify $b$ to $\hat b$: in each (undirected)
component $K$ of $F$, pick a node $j\in K$ and change $b_j$ by
$-\sum_{i\in K} b_{i}$; leave all other $b_i$ values unchanged. Thus
we get a $\hat b$ with $D_{\hat b}(F)=0$.  
\textsc{Trial}$(F,\hat b)$ returns an $F$-tight vector $\hat f$.

{\bf Step 2.} Call the subroutine \textsc{Error}$(\hat f,F)$.
 If $b=\hat b$ and $err_F(\hat f)=0$, then $\hat f$ is $F$-optimal.
An optimal solution to (\ref{probl})
can be found by a single maximum flow computation, as described in the
proof of  Lemma~\ref{lem:compute-opt}.
 In this case, the algorithm terminates.
If $err_F(\hat f)\ge \Delta/2$, then keep the original $f$, perform 
\textsc{Adjust}$(\Delta/2,f,F)$  and go to the next scaling phase
with scaling factor $\Delta/2$. Otherwise, set $f=\hat f$ and define
the next scaling factor as
\begin{equation*}
\Delta_{next}=\max\{err_F(\hat f),Ex_b(\hat f)/(2n+m_N)\}. 
\end{equation*}

\section{Analysis}\label{sec:analysis}
The details how the shortest path computations are performed will be discussed in Section~\ref{sec:shortest}; 
in the following analysis, we assume it can be efficiently implemented.
 At the initialization, 
$err_\emptyset({\bf 0})$ must be finite or the problem is unbounded
as assumed in Oracle~\ref{assump:error}. 

\textsc{Trial-and-Error}
replaces $f$ by $\hat f$ if  $err_F(\hat f)\le \Delta/2$ and keeps the
same $f$ otherwise. The first case will be called a {\sl successful}
trial, the latter is {\sl unsuccessful}. 
The following is (an almost identical) counterpart of
Lemma~\ref{lem:basic-alg}.
\begin{lemma}\label{lem:enhanced}
\begin{enumerate}[(i)]
\item In the main part of the $\Delta$-phase, the $F$-pseudoflow $f$ is an integer
  multiple of $\Delta$ on each arc $ij\in E\setminus F$, and consequently,
  $E_f^F(\Delta)=E_f^F$.
\item $(\Delta,F)$-feasibility is maintained in the main part and in
subroutine \textsc{Extend}$(\Delta,f,F)$. 
\item
 At the beginning of the main part, $Ex(f)\le (2n+m_N)\Delta$, and at
 the end, $Ex(f)\le n\Delta$. 
\item The main part consists of at most
$2n+m_N$ flow augmentation steps.
\item The scaling factor $\Delta$ decreases by at least a factor of 2 between two $\Delta$-phases. 
\end{enumerate}
\end{lemma}
\begin{proof}
For {\em (i)}, $f$ is zero on every arc in $E\setminus F$ at the beginning of the algorithm
and after every successful trial. In every other case, the previous
phase had scaling factor $2\Delta$, and thus by induction, the flow is an integer
multiple of $2\Delta$ at the end of the main part of the $2\Delta$-phase, a property also maintained by \textsc{Extend}$(2\Delta,f,F)$. 
The $2\Delta$-phase finishes with \textsc{Adjust}$(\Delta,f,F)$, possibly modifying the flow on every arc by $\pm \Delta$.
In the main part of the $\Delta$-phase,
the shortest path augmentations also change the flow by
$\pm \Delta$. This implies $E_f^F(\Delta)=E_f^F$.

For {\em (ii)}, $P$ is a shortest path if there exists a potential $\pi$
satisfying (\ref{eq:delta-F-feas}) 
 with $\pi_j-\pi_i=C'_{ij}(f_{ij}+\Delta)$ on each arc $ij\in
P$ (see also Section~\ref{sec:shortest}). We show that when augmenting on the shortest path $P$, 
(\ref{eq:delta-F-feas}) is maintained with the same
$\pi$.
If $ij,ji\notin P$, then it is trivial as the flow is left unchanged on $ij$. Consider now an arc $ij\in P$; the next argument applies both if $ij$ is a forward or a reverse arc. The new flow value will be
$f_{ij}+\Delta$, hence we need $\pi_j-\pi_i\le
C'_{ij}(f_{ij}+2\Delta)$, obvious as $C'_{ij}$ is monotonely
increasing. We next verify  (\ref{eq:delta-F-feas})  for the backward arc $ji\in E_f^F(\Delta)$.
This gives  $\pi_i-\pi_j\le C'_{ji}((f_{ji}-\Delta)+\Delta)$, that is equivalent to  $C'_{ij}(f_{ij})\le \pi_j-\pi_i$, again a consequence of
 monotonicity.


In subroutine \textsc{Extend}, we reroute the flow $f_{ij}$ from a
linear arc $ij$ if $\ole F$ contains a directed path $P$ from $i$ to
$j$. This cannot affect feasibility since the $C'_{ij}$'s are constant
on linear arcs. Also note that arcs in $\ole F$ have infinite residual capacities.

For {\em (iii)}, $Ex(f)\le n\Delta$ as the main part terminates with either
$S(\Delta)=\emptyset$ or $T(\Delta)=\emptyset$. Lemma~\ref{lem:adjust-F} shows that
\textsc{Adjust}$(\Delta/2,f,F)$ increases the excess by at most $m_N\Delta/2$.
Consequently, $Ex(f)\le (2n+m_N)(\Delta/2)$ at the beginning of the $\Delta/2$-phase.

The other possible case is that a successful trial replaces $\Delta$ by
$\Delta_{next}$. By  definition, the new excess is at most
$(2n+m_N)\Delta_{next}$.

Further, {\em (iii)} implies {\em (iv)}, as each flow
augmentation decreases $Ex(f)$ by $\Delta$. Finally {\em (v)} is straightforward if the next value of the scaling factor is set as $\Delta/2$.
This is always the case, except if \textsc{Trial-and-Error} is called and
 $err_F(\hat f)\le \Delta/2$, when the next scaling factor is set as the maximum of $err_F(\hat f)$ and
$Ex_b(\hat f)/(2n+m_N)$. We show that this second term is also at most $\Delta/2$. Indeed, $\hat f$ was obtained by \textsc{Trial}$(F,\hat b)$,
and therefore $\rho_{\hat f}(i)-b_i=\hat b_i - b_i\le \Delta$ due to the definition of $\hat b$ and $D_b(F)\le \Delta$. It follows that $Ex_b(\hat f)\le n\Delta$, and thus $Ex_b(\hat f)/(2n+m_N)<\Delta/2$. 
\end{proof}

\begin{lemma}\label{lem:F-in-F}
 $F\subseteq F^*$ holds in each step of the algorithm.
\end{lemma}
\begin{proof}
The proof is by induction. A new arc $ij$ may enter $F$ if
$f_{ij}>(2n+m+1)\Delta$ after the main part of the $\Delta$-phase.
We shall prove that $f^*_{ij}>0$ for some $F$-optimal solution $f^*$, and thus
 Lemma~\ref{lem:compute-opt} gives $ij\in F^*$.

After the phase when $ij$ entered, let us continue with the following modified algorithm:
do not extend $F$ and do not perform
\textsc{Trial-and-Error} anymore, but 
always choose the next scaling
factor as $\Delta/2$, and keep the algorithm running forever. 
(This is almost the same as the Basic algorithm, with the difference that we have a revealed arc set $F$.)

Let $\Delta_0=\Delta$ and $\Delta_t=\Delta/2^t$ denote the scaling factor in the $t$'th phase of this algorithm (with phase 0 corresponding to the $\Delta$-phase).
Consider any $\Delta_t$-phase ($t\ge 1$).
The flow is modified by at
most $(2n+m_N)\Delta_t$ during the main
part by Lemma~\ref{lem:enhanced}(iv) and by $\Delta_t/2$ in \textsc{Adjust}$(\Delta_t/2,f,F)$, amounting to a total
modification $\le (2n+m_N+\frac12)\Delta_t$. Consequently, 
the total modification in the $\Delta_t$ phase and all later phases is bounded by
$(2n+m_N+\frac12)\sum_{k=t}^\infty \Delta_k\le 2(2n+m+\frac12)\Delta_t$.

We may conclude that when running forever, the flow $f$ converges to an $F$-optimal solution $f^*$.
Indeed, let $f^{(t)}$ denote the $F$-pseudoflow at the end of the $t$'th phase. By the above observation, 
$||f^{(t)}-f^{(t')}||_\infty\le 2(2n+m+\frac12)\Delta_t$ for any $t'\ge t\ge 0$. Consequently, on every arc $ij\in E$, the sequence
$f^{(t)}_{ij}$ converges; let $f^*$ denote the limit. We claim the $f^*$ is $F$-optimal.

Firstly, $f^*$ is clearly an $F$-pseudoflow. Property
(\ref{cond:F-opt}) is equivalent to the property that $E_f^F$ does not contain any negative cycle w.r.t. $C'_{ij}(f_{ij})$.
This follows from the fact that $E_f^F(\Delta_t)$ does not contain any negative cycle w.r.t. $C'_{ij}(f_{ij}^{(t)})$ due to the
$(\Delta_t,F)$-feasibility of $f^{(t)}$.
Finally, $Ex_b(f^*)=\lim_{t\rightarrow\infty} Ex_b(f^{(t)})\le \lim_{t\rightarrow\infty} n\Delta^t=0$, and therefore 
$Ex_b(f^*)=0$.

To finish the proof, we observe that $f^*_{ij}>0$. Indeed, $f_{ij}>(2n+m+1)\Delta$ after the main part of the $\Delta$-phase, and hence
 $f_{ij}>(2n+m+\frac 12)\Delta$ at the end of the $\Delta$-phase (after performing  \textsc{Adjust}$(\Delta/2,f,F)$). By the above argument,
 the total change in all later phases is $\le 2(2n+m+\frac12)\Delta_1=(2n+m+\frac12)\Delta$, yielding the desired conclusion.
\end{proof}

Recall the characterization of arcs to free and restricted. Free arcs are differentiable on the entire $\R$, whereas for a restricted arc $ij$, we have  $C'_{ij}(\alpha)=-\infty$ for $\alpha< 0$. Therefore we have to avoid the flow value becoming negative even if $ij\in F$ for a restricted arc.

\begin{claim}\label{cl:restricted-arcs}
$f_{ij}\ge 0$ holds for every restricted arc $ij$ during the entire algorithm, even if $ij\in F$.
\end{claim}
\begin{proof}
 $f_{ij}\ge 0$ holds at the initialization; consider the first $\Delta$-phase  when $f_{ij}<0$ is attained. This can happen 
during a path augmentation or in the \textsc{Adjust} subroutine (\textsc{Extend} may not modify $f_{ij}$ as $ij$ is a nonlinear arc).
In case of a path augmentation, $ji$ is contained on the shortest path $P$, and therefore $\pi_j-\pi_i=C'_{ij}(f_{ij}-\Delta)$ must hold for a potential $\pi$ (see the proof of Lemma~\ref{lem:enhanced}). This is a contradiction as $f_{ij}-\Delta<0$ and thus $C'_{ij}(f_{ij}-\Delta)=-\infty$. A similar argument works for 
\textsc{Adjust}. 
\end{proof}

\begin{lemma}\label{lem:error-bound} 
When \textsc{Trial-and-Error}$(F)$ is performed in the $\Delta$-phase,
$err_F(\hat f)\le 6(m+1)^2\Delta$ holds.
\end{lemma}
This lemma is of key importance. Before proving it, we show how it provides the strongly
polynomial bound. The main idea is the following: in \textsc{Trial-and-Error}$(F)$, we replace $f$ by $\hat f$ and $\Delta$ by a new value instead of  $\Delta/2$  in case $err_F(\hat f)<\Delta/2$; otherwise, we ignore $\hat f$ and proceed to the next phase as usual. Whereas $err_F(\hat f)\ge \Delta/2$ is possible, the lemma gives an upper bound in terms of $\Delta$. Note also that the output of the subroutine \textsc{Trial-and-Error}$(F)$ depends only on the revealed arc set $F$.
Consequently, if we had $err_F(\hat f)\ge \Delta/2$, then by the time the scaling factor reduces to a smaller value $\Delta'$ such that $6(m+1)^2\Delta'<\Delta/2$, the set $F$ must have been extended.
\begin{theorem}\label{thm:main-bound}
The enhanced algorithm terminates in at most  $O((n+m_N)\log m)$ scaling phases.
\end{theorem}
\begin{proof}
The set of revealed arcs can be extended at most $m_N+n-1$
times, since there can be at most $(n-1)$ linear arcs because of the
linear acyclic property. We shall  show that after any $\Delta$-phase, a new
arc is revealed within $2\lceil\log_2 T\rceil$ phases, for $T=24(m+1)^2$.

As $\Delta$ decreases by at least a
factor of two between two phases, after $\lceil\log_2 T\rceil$ steps we  have
$\Delta_T\le {\Delta}/{T}$. Assume that in the
$\Delta_T$ phase,  we still have the same revealed arc set $F$ as in the $\Delta$-phase.

\medskip
\noindent{\bf Case I.
 $D_b(F)>\Delta$.} At the end of the main part of the $\Delta_T$-phase, 
$D_b(F)> 24(m+1)^2\Delta_T$. Thus there is an undirected connected component $K$ of
$F$ with $|\sum_{i\in K}b_i|>24(m+1)^2\Delta_T$.
Let $\rho_f(K)$ denote the total $f$ value on arcs entering $K$ minus the value on arcs leaving $K$, that is,
\[
\rho_f(K):=\sum_{ij\in E:i\notin K,j\in K}f_{ij}-\sum_{ij\in E:i\in K,j\notin K}f_{ij}.
\]
We have
\[
|\rho_f(K)|=\left|\sum_{i\in K}\rho_f(i)\right|=\left|\sum_{i\in K}(\rho_f(i)-b_i+b_i)\right|\ge \left|\sum_{i\in K}b_i\right|-Ex_b(f).
\]
The last part is derived from the simple inequality $|\beta+\alpha^++\alpha^-|\ge |\beta|-\gamma$, whenever
$\alpha^+,\alpha^-,\beta,\gamma\in\R$ with $-\gamma\le\alpha^-\le0\le\alpha^+\le\gamma$. In our setting, 
$\beta=\sum_{i\in K}b_i$, $\alpha^+=\sum_{i\in K}\max\{\rho_f(i)-b_i,0\}$, $\alpha^-=\sum_{i\in K}\min\{\rho_f(i)-b_i,0\}$, and
$\gamma=Ex_b(f)$. The conditions hold since 
\[\gamma=Ex_b(f)=\sum_{i\in V}\max\{\rho_f(i)-b_i,0\}=-\sum_{i\in V}\min\{\rho_f(i)-b_i,0\}.\]
For the second equality, note that $\sum_{i\in V}b_i=\sum_{i\in V}\rho_f(i)=0$.
Now we may conclude
\[
\left|\rho_f(K)\right|\ge \left|\sum_{i\in K}b_i\right|-Ex_b(f)> 24(m+1)^2\Delta_T-n\Delta_T> (2n+m+1)m\Delta_T.
\]
Consequently, there must be an arc $ij$ entering or leaving $K$ with
$f_{ij}> (2n+m+1)\Delta_T$, a contradiction as at least one such arc must have been
added to $F$ in \textsc{Extend}$(\Delta_T,f,F)$. Note that the first such arc examined
during  \textsc{Extend}$(\Delta_T,f,F)$ does keep the linear acyclic property as it connects two separate connected components of $F$.

\medskip
\noindent{\bf Case II.
 $D_b(F)\le\Delta$.} 
 We may assume that either we are at
the very beginning of the algorithm with $F=\emptyset$, or in a
phase when
$F$ just has been
extended; otherwise, we could
consider an earlier phase with this property. We can interpret the
initial solution $\bf 0$ and initial $\Delta$ as the output of
\textsc{Trial-and-Error}$(\emptyset)$.

\smallskip
{\bf Case IIa. $D_b(F)>\Delta_T$.}  The argument of Case I, applied for $\Delta_T$ instead of $\Delta$, shows that within 
$\lceil\log_2 T\rceil$ phases after the $\Delta_T$ phase, $F$ shall be extended, showing that a new  arc was revealed within $2\lceil\log_2 T\rceil$ phases after the $\Delta$-phase.

\smallskip
{\bf Case IIb. $D_b(F)\le\Delta_T$.}
Recall the  assumption that $F$ has not changed between phases $\Delta$ and $\Delta_T$, and thus $D_b(F)$ has not changed its value either.
Let us apply the
analysis of the \textsc{Trial-and-Error} subroutine for the $\Delta_T$-phase.
(Even if the subroutine is not actually performed, its analysis is
valid provided that $D_b(F)\le \Delta_T$.)

Let $\hat f$ be the arc set found by \textsc{Trial}$(F,\hat b)$. Let us assume
that $b$ is modified to
$\hat b$ always the same way for the same $F$; with this assumption, the output of the subroutine is the same
whether called in the $\Delta$ or in the $\Delta_T$-phase.
In the event of an unsuccessful trial in the $\Delta$-phase,
$\Delta/2\le err_F(\hat f)$. Using
Lemma~\ref{lem:error-bound} for the $\Delta_T$-phase,
\[
err_F(\hat f)\le 6(m+1)^2\Delta_T\le \Delta/4\le err_F(\hat f)/2,
\]
a contradiction. On the other hand, if we had a successful trial in the $\Delta$-phase, then
$\Delta_T\le 2\Delta_{next}/T$, as $\Delta_T$ is the scaling factor  $T-1$ phases after the $\Delta_{next}$-phase.
 Lemma~\ref{lem:error-bound} and  $Ex_b(\hat f)\le nD_b(F)\le n\Delta_T$ together yield
\[
\Delta_{next}=\max\{err_F(\hat f),Ex_b(\hat f)/(2n+m_N)\}\le 6(m+1)^2\Delta_T\le
\Delta_{next}/2,
\]
a contradiction again.
\end{proof}

Some preparation is needed to prove Lemma~\ref{lem:error-bound}. We note that the linear acyclic property is important due to the following lemma; if $F$ may contains undirected cycles of linear arcs, the claim is not true.
\begin{lemma}\label{lem:f1f2}
For a linear acylic arc set $F\subseteq E$, let $x$ and $y$ be two $F$-tight vectors.
Then 
$||x-y||_\infty\le ||\rho_{x}-\rho_{y}||_1$ holds.
\end{lemma}
\begin{proof}
First, we claim that the difference graph $D_{x,y}=(V,E_{x,y})$ is
acyclic.
Indeed, if there existed a cycle $C\subseteq E_{x,y}$, then
we get $0=\sum_{ab\in C}C'_{ab}(x_{ab})=\sum_{ab\in
  C}C'_{ab}(y_{ab})$ as in the proof of Lemma~\ref{lem:canon-F}. Since $x_{ab}>y_{ab}$ for every $ab\in C$, this is only possible if all
arcs of $C$ are linear (\ref{assump:lin-nonlin-igazi}), contradicting the linear acyclic property of
$F$. (Note that $E_{x,y}\subseteq \ole F$, since by definition, every $F$-tight vector is supported on $F$).

Define the function $z$ by $z_{ij}=x_{ij}-y_{ij}>0$ for $ij\in
E_{x,y}$ (again with the convention $x_{ji}=-x_{ij}$, $y_{ji}=-y_{ij}$ if
$ij\in E$). $\rho_z\equiv \rho_x-\rho_y$, therefore we have to
prove $z_{ij}\le ||\rho_z||_1$ for $ij\in E_{x,y}$. This property
indeed holds for every positive $z$ with acyclic support.

Consider a reverse topological ordering $v_1,\ldots,v_n$ of
$V$, where $v_av_b\in E_{x,y}$ implies $a>b$. For the arc $ij\in E_{x,y}$,
let $i=v_{t'}$ and $j=v_{t}$ ($t'>t$).
Let $V_t=\{v_1,\ldots,v_t\}$. $V_t$ is a directed cut in $E_{x,y}$, thus
\[
\sum_{p>t\ge q}z_{v_pv_q}=\sum_{p\le
  t}\rho_z(v_p).
\]
As $z$ is positive on all arcs, this implies $z_{v_av_b}\le
\sum_{p\le
  t}\rho_z(v_p)\le ||\rho_{z}||_1$ for all such arcs, in particular, for $ij$.
\end{proof}

\begin{claim}\label{cl:feas-bound}
If $f$ and $\hat f$ are $F$-pseudoflows with $\hat f_{ij}=0$ for
$ij\in E\setminus F$, and $f$ is
$(\Delta,F)$-feasible,
then $\hat f$ is  $(\Delta+||f-\hat f||_\infty,F)$-feasible.
\end{claim}
\begin{proof}
There is a potential $\pi$ so that $f$
and $\pi$ satisfy (\ref{eq:delta-F-feas}), that is,
$\pi_j-\pi_i\le C'_{ij}(f_{ij}+\Delta)$ if $ij\in E^F_f(\Delta)$.
For $\alpha=||f-\hat f||_\infty$, we have $f_{ij}+\Delta\le \hat f_{ij}+\Delta+\alpha$. Consequently,
(\ref{eq:delta-F-feas}) is satisfied for $(\hat {f}_{ij},\pi)$ and
$\Delta+\alpha$
for every arc in $E^F_f(\Delta)$.

By the assumption that $\hat f$ is zero outside $F$, we have
$E^F_{\hat f}(\Delta+\alpha)=E\cup\obe F\subseteq E^F_f(\Delta)$ and thus the
claim follows.
\end{proof}

\begin{proof}[Proof of Lemma~\ref{lem:error-bound}]
When \textsc{Trial-and-Error} is applied,
$f$ is $(\Delta,F)$-feasible with some potential $\pi$ and $Ex_b(f)\le n\Delta$.
We claim that there is an $F$-tight $\bar f$ so that 
$|\bar f_{ij}-f_{ij}|\le \Delta$ for every $ij\in F$, and $Ex_b(\bar f)\le
(2n+m+2)m\Delta$.

Indeed, $(\Delta,F)$-feasibility gives
\[
C'_{ij}(f_{ij}-\Delta)\le \pi_j-\pi_i\le C'_{ij}(f_{ij}+\Delta)\ \
\forall ij\in F.
\]
If $ij$ is a free arc (that is, differentiable on the entire $\R$), then 
 $C'_{ij}$ is continuous, so there must be a value $f_{ij}-\Delta\le
\beta\le f_{ij}+\Delta$ with $C'_{ij}(\beta)=\pi_j-\pi_i$. 
This also holds if $ij$ is a restricted arc, since by Claim~\ref{cl:restricted-arcs}, 
$f_{ij}\ge 0$ and $C'_{ij}$ is continuous on $(\max\{0,f_{ij}-\Delta\},f_{ij}+\Delta)$, and $C'_{ij}(0)=-\infty$.
Let us set
$\bar f_{ij}=\beta$. 
This increases $Ex_b(f)$ by at most $|F|\Delta$. 

Let us set
$\bar f_{ij}=0$ for $ij\in E\setminus F$. Note that $f_{ij}\le (2n+m+1)\Delta$
if $ij\notin F$ (every arc with $f_{ij}>(2n+m+1)\Delta$ is either added to $F$ or is modified to $f_{ij}=0$ in the subroutine \textsc{Extend}).
Further, $Ex_b(f)\le n\Delta$, and thus we obtain an
$F$-tight $\bar f$ with 
\begin{align*}
Ex_b(\bar f)\le n\Delta+|F|\Delta+(2n+m+1)(m-|F|)\Delta\\
\le (2n+m+2)m\Delta.
\end{align*}
On the other hand,
$Ex_b(\hat f)\le nD_b(F)\le n\Delta$, since $Ex_{\hat b}(\hat f)=0$ and $\hat b$ is obtained from $b$ by modifying certain values by $\le D_b(F)$.
Consequently,
\[
||\rho_{\bar f}-\rho_{\hat f}||_1\le
||\rho_{\bar f}-b||_1+||\rho_{\hat f}-b||_1=2Ex_b(\bar f)+2Ex_b(\hat f)\le
2(2n+m+3)m\Delta\le 6m(m+1)\Delta.
\]

Applying Lemma~\ref{lem:f1f2} for $x=\bar f$ and $y=\hat f$ gives
$||\hat f-\bar f||_\infty\le 6m(m+1)\Delta$.
We also have $||f-\bar f||_\infty\le (2n+m+1)\Delta\le (3m+1)\Delta$ by the construction, and therefore
\[
||f-\hat f||_\infty\le ||f-\bar f||_\infty+||\bar f-\hat f||_\infty< 6(m+1)^2\Delta-\Delta
\]
Applying  Claim~\ref{cl:feas-bound} for $f$ and $\hat f$ we conclude that $\hat f$ is $6(m+1)^2\Delta$-feasible; recall that $f$ was
$(\Delta,F)$ feasible when we applied \textsc{Trial-and-Error}.
\end{proof}
\begin{theorem}\label{thm:running-time-bound}
Let $\rho_S(n,m)$ be the running time needed for one shortest path
computation for nonnegative lengths. Then the running time of the algorithm is bounded by 
\[
O((n+m_N)(\rho_T(n,m)+\rho_E(n,m))+(n+m_N)^2\rho_S(n,m) \log m).
\]
\end{theorem}
\begin{proof}
By Theorem~\ref{thm:main-bound}, there are at most $(n+m_N)\log m$ scaling
phases, each dominated by $O(n+m_N)$ shortest path computations. The
subroutine \textsc{Trial-and-Error} is performed only when $F$ is
extended, that is, at most $n+m_N$ times, and comprises the subroutines
\textsc{Trial} and \textsc{Error}.
\end{proof}


\subsection{Shortest path computations}\label{sec:shortest}
For the sake of efficiency, we shall maintain a potential vector $\pi$ during the entire algorithm such that
$(f,\pi)$ satisfies the condition (\ref{eq:delta-F-feas}) on $(\Delta,F)$-feasibility. 

For the initial $\Delta$ value, $\Delta\ge err_\emptyset({\bf 0})$, and the latter value is computed by \textsc{Error}$({\bf 0},\emptyset)$.
This means that $f={\bf 0}$  is $(\Delta,\emptyset)$-feasible. Similarly, after every successful trial we have a new flow $\hat f$ computed by \textsc{Error}$(f,F)$ and new scaling factor value  
$\Delta_{next}\ge err_F(\hat f)$.  In the applications, this subroutine will also return a potential vector $\pi$ such that $(f,\pi)$ satisfies (\ref{eq:delta-F-feas}).

Alternatively, such a potential vector may be obtained by the standard label correcting
algorithm (see \cite[Chapter 5.5]{amo}), since it is a dual proof of the fact that the graph $E_f^F(\Delta)$ contains no negative cycles with respect to the cost function $C'_{ij}(f_{ij}+\Delta)$; we have access to these values via the value oracle (Oracle~\ref{assump:oracle}).

In the main part of the $\Delta$-phase, we may apply a variant of Dijkstra's
algorithm (see \cite[Chapter 4.5]{amo}) to compute shortest
paths. This needs a nonnegative cost function, but instead of the
original $C'_{ij}(f_{ij}+\Delta)$ that may take negative values, we shall use 
$C'_{ij}(f_{ij}+\Delta)-\pi_j+\pi_i$, a nonnegative function by
(\ref{eq:delta-F-feas}); the set of shortest paths is identical  for the two costs.
This subroutine can be implemented by updating the potentials $\pi$, so that  $(\Delta,F)$-feasibility is maintained, and we obtain $C'_{ij}(f_{ij}+\Delta)=\pi_j-\pi_i$ on every arc of 
every shortest path. For the sake of completeness, we describe this subroutine in the Appendix.

As shown in the proof of Lemma~\ref{lem:enhanced}(ii), once we have a potential $\pi$ such that $C'_{ij}(f_{ij}+\Delta)=\pi_j-\pi_i$ on every arc of 
a shortest path $P$, then sending $\Delta$-units of flow on $P$ maintains (\ref{eq:delta-F-feas}) for $(f,\pi)$. It is also maintained in \textsc{Extend}$(\Delta,f,F)$ since flow values are modified only on arcs with $C'_{ij}$ constant. Finally,
\textsc{Adjust}$(\Delta/2,f,F)$ modifies the flow so that   (\ref{eq:delta-F-feas}) is maintained for the same $\pi$ and $\Delta/2$ by Lemma~\ref{lem:adjust-F}.

Let us now explore the relation to Oracle~\ref{assump:oracle}.
In both applications, we shall verify that the subroutine \textsc{Trial-and-error} returns a rational flow vector $f$ and a rational value $\Delta$.
Since flow will always be modified in units of $\Delta$ in all other parts of the algorithm, we may conclude that a rational $f$ will be maintained in all other parts. Under 
Oracle~\ref{assump:oracle}(a) (i.e., quadratic objectives), we shall maintain a rational potential vector $\pi$,
while under Oracle~\ref{assump:oracle}(b) (i.e., Fisher markets), we shall maintain the rationality of the $e^{\pi_i}$ values;
during the computations, we shall use the representation of these values instead of the original $\pi$. For this aim, we will use a multiplicative
variant of Dijkstra's algorithm, also described in the Appendix. We shall also verify that in the corresponding applications, 
the subroutine \textsc{Error}$(f,F)$ returns a potential vector $\pi$ so that  $(f,\pi)$ satisfies (\ref{eq:delta-F-feas}), with the $\pi_i$ or the $e^{\pi_i}$ values being rational, respectively.

Finally, it is easy to verify that whereas we are working on a transformed uncapacitated instance, we may use the complexity bound of the original instance, as summarized in the following remark.
\begin{remark}
\label{remark:shortest}
{\em
A shortest path computation can be performed in time
$\rho_S(n,m)=O(m+n\log n)$ using Fibonacci heaps, see \cite{Fredman87}. Recall that the
original problem instance was on $n'$ nodes and $m'$ arcs, and it was
transformed to an uncapacitated instance on $n=n'+m'$ nodes and $m=2m'$
arcs. However, as in Orlin's algorithm \cite{Orlin93}, we can use the
bound $O(m'+n'\log n')$ instead of $O(m'+m'\log n')$ because shortest
path computations can be essentially performed on the original network.}
\end{remark}

\section{Applications}\label{sec:applications}
\subsection{Quadratic convex costs}\label{sec:quad}
Assume that $C_{ij}(\alpha)=c_{ij}\alpha^2+d_{ij}\alpha$ for each $ij\in E$,
with $c_{ij}\ge 0$. This clearly satisfies the assumption in
Oracle~\ref{assump:oracle}(i) since
$C'_{ij}(\alpha)=2c_{ij}\alpha+d_{ij}$. Also,
(\ref{assump:lin-nonlin-igazi})
is satisfied: $ij$ is linear if $c_{ij}=0$.

The subroutine \textsc{Trial}$(F,b)$ can be implemented by solving a system
of linear equations. 
\begin{align}
\pi_j-\pi_i&=2c_{ij}x_{ij}+d_{ij} \quad\forall ij\in F\notag \\
\sum_{j:ji\in F} x_{ji}-\sum_{j:ij\in F} x_{ij}&=b_i  \quad\forall
i\in V\label{sys:eq}\\
x_{ij}&=0\quad \forall ij\in E\setminus F \notag
\end{align}
The conditions in Oracle~\ref{assump:trial} is verified by the next claim.

\begin{lemma}\label{lem:quad-solve}
Let $F$ be linear acyclic (that is, there is no undirected cycle of arcs with $c_{ij}=0$) with $D_b(F)=0$. 
Then (\ref{sys:eq}) is feasible and a solution can be
found in $\rho_T(n,m)=O(n^{2.37}+m)$ time.
\end{lemma}

\begin{proof}
Clearly, we can solve the system separately on different undirected connected
components of $F$. In the sequel, let us focus on a single connected
component; for simplicity of notation, assume this component is the entire $V$.

Consider first the case when  all arcs are linear. Then we can solve
the equalities corresponding to edges and nodes separately. 
As $F$ is assumed to be linear acyclic, it forms a tree. If we fix one $\pi_j$ value
arbitrarily, it determines all other $\pi_i$ values  by
moving along the edges in the tree. The $x_{ij}$'s can be found by
solving a flow problem on the same tree with the demands $b_i$. This is clearly feasible
by the assumption $D_b(F)=0$, that is, $\sum_{i\in V} b_i=0$ (note that we do not have nonnegativity constraints on the arcs). Both
tasks can be performed in linear time.

Assume next both linear and nonlinear arcs are present, and let $T$ be
an undirected connected component of linear arcs. As above, all $\pi_j-\pi_i$
values for $i,j\in T$ are uniquely determined. If there is a nonlinear
arc $ij\in F$
with $i,j\in T$, then $x_{ij}=(\pi_j-\pi_i-d_{ij})/(2c_{ij})=\alpha$ is also uniquely determined. We
can remove this edge by replacing $b_i$ by $b_i+\alpha$ and $b_j$
by $b_j-\alpha$. 
Hence we may assume that the components of linear arcs span no
nonlinear arcs. 

Next, we can contract each such component $T$ to a single
node $t$ by setting $b_t=\sum_{i\in T}b_i$ and
modifying the $d_{ij}$ values on incident arcs as follows. 
Let $t$ correspond to a fixed node in $T$, and consider  an arc with $i\in T$, $j\notin T$.
Let $\alpha$ denote the sum of $d_{ab}$ values on the $t-i$ path in $T$; let us  add $\alpha$ to $d_{ij}$.
Similarly for an arc $ij$ entering $T$ we must subtract the sum of the costs on the $t-j$ path from $d_{ij}$.
A
solution to the contracted problem can be easily extended to the
original instance.

For the rest, we can assume all arcs are nonlinear, that is,
$c_{ij}>0$ for all $ij\in F$. 
Let $A$ be the node-arc incidence matrix of $F$:
$A_{i,ij}=-1$, $A_{i,ji}=1$ for all $ij\in F$, and all other entries
are 0.  Let $C$ be the $|F|\times|F|$ diagonal matrix with
$C_{ij,ij}=-2c_{ij}$. (\ref{sys:eq}) can be written in the form
\[
\left(
\begin{array}{cc}
A^T& C\\
0 & A
\end{array}
 \right) (\pi,x)^T=
\left(
\begin{array}{c}
d\\
b
\end{array}
\right).
\]
This can be transformed into 
\[
\left(
\begin{array}{cc}
A^T& C\\
L & 0
\end{array}
 \right)(\pi,x)^T =
\left(
\begin{array}{c}
d\\
b'
\end{array}
\right),
\]
where $L$ is the weighted Laplacian matrix with $L_{ii}=\sum_{j:ij\in
  \ole F}\frac{1}{2c_{ij}}$, $L_{ij}=L_{ji}=-\frac{1}{2c_{ij}}$ if $ij\in
F$ and $L_{ij}=0$ otherwise, and $b'$ is an appropriate vector with
$\sum_{i\in V}b'_i=0$.

The main task is to solve the system $L\pi=b'$. It is well-know (recall that $V$ is assumed to be a single
connected component) that $L$ has rank $|V|-1$ and the system
is always feasible whenever $\sum_{i\in V}b'_i=0$.
A solution can be found in  $O(n^{2.37})$ time \cite{Coppersmith90}. All
previously described operations (eliminating nonlinear arcs spanned in components of linear arcs,
contracting components of linear arcs) can be done in $O(m)$ time, hence 
the bound $\rho_T(n,m)=O(n^{2.37}+m)$. 
%
\end{proof}

 To implement \textsc{Error}$(f,F)$, we have an $F$-tight vector $f$, and we need 
to find the minimum
 $\Delta$-value such that there exists a $\pi$ potential with
\begin{equation}
\pi_j-\pi_i\le (2c_{ij}f_{ij}+d_{ij})+2c_{ij}\Delta \ \ \forall ij\in
E\cup \overleftarrow F.\label{eq:error}
\end{equation}
We show that this can be reduced to the minimum-cost-to-time ratio
cycle problem, defined as follows (see \cite[Chapter
5.7]{amo}). In a directed graph, there is a cost function $p_{ij}$ and a
time $\tau_{ij}\ge 0$ associated with each arc. The aim is to find a
cycle $C$ minimizing $(\sum_{ij\in C} p_{ij})/(\sum_{ij\in C}
\tau_{ij})$. A strongly polynomial algorithm was given by Megiddo
\cite{Megiddo79,Megiddo83}
that solves the problem in $\min\{O(n^3\log^2
n),$ $O(n\log n (n^2+m\log\log n))\}$ time.
The problem can be equivalently formulated as
\begin{align}\label{eq:param}
\min \mu \mbox{ s. t. there are no negative cycles}\notag\\
\mbox{ for the cost function }p_{ij}+\mu\tau_{ij}.
\end{align}
Our problem fits into this framework with
$p_{ij}=2c_{ij}f_{ij}+d_{ij}$ and $\tau_{ij}=2c_{ij}$. In 
(\ref{eq:param}), the optimal $\mu$ value is $-\Delta$.
 However, \cite{Megiddo79} defines the minimum ratio cycle problem
 with $\tau_{ij}>0$ for every $ij\in E$.
This property is not essential for Megiddo's algorithm, which uses a
parametric search method for $\mu$ to solve 
 (\ref{eq:param}) under the
only (implicit) restriction that the problem is feasible. 

In our setting $\tau_{ij}>0$ holds for nonlinear arcs, but
$\tau_{ij}=0$ for linear arcs.
Also, there can be cycles $C$ with $\sum_{ij\in
  C}\tau_{ij}=0$. (This can happen even if $F$ is linear acyclic, as
$C$ can be any  cycle in $E\cup\obe F$.) If we have such a cycle $C$
with $\sum_{ij\in  C}p_{ij}<0$, then (\ref{eq:param}) is
infeasible. In every other case, the problem is feasible and thus
Megiddo's algorithm can be applied.


For this reason, we first check whether there is a negative cycle with
respect to the $p_{ij}$'s in  the set of linear arcs in
$E\cup\overleftarrow F$. This can be done via the label correcting algorithm in $O(nm)$ time (\cite[Chapter 5.5]{amo}).
If there exists one, then (\ref{eq:error}) is infeasible, thus 
$err_F(f)=\Delta=\infty$, and (\ref{probl}) is
unbounded as we can send arbitrary flow around this cycle.
Otherwise, we have  $\sum_{ij\in C}\tau_{ij}>0$ for every cycle with
$\sum_{ij\in C}p_{ij}<0$, and consequently, there exists a finite $\Delta$
satisfying (\ref{eq:error}).

Consequently, $\rho_T(n,m)=\min\{O(n^3\log^2
n), O(n\log n (n^2 +m\log\log n))\}$.
Theorem~\ref{thm:running-time-bound} gives the following running time bound.
\begin{theorem}\label{thm:quadratic-run}
For convex quadratic objectives on an uncapacitated instance on $n$
nodes and $m$ arcs, the algorithm finds an optimal solution in
$O(m(n^3\log^2 n+ m\log m(m+n\log n)))$ time. For a capacitated instance, the running time can be
bounded by $O(m^4\log m)$.
\end{theorem}
The bottleneck is clearly the $m$ minimum-cost-to-time
computations. 
As in Remark~\ref{remark:shortest}, it is
likely that one can get the same running time $O(m(n^3\log^2 n+ m\log m(m+n\log n)))$ for
capacitated instances via a deeper analysis of Megiddo's algorithm.

\medskip

Let us verify that the algorithm is strongly polynomial. It uses elementary arithmetic operations only, and the running time is polynomial in $n$ and $m$, according to the above theorem. It is left to verify requirement (iii) on strongly polynomial algorithms (see the Introduction): if all numbers in the input are rational,
then every number occurring in the computations is rational and is of size polynomially bounded in the size of the input.

At the initialization and in every successful trial, we compute a new flow $f$ by solving (\ref{sys:eq}) as described in
Lemma~\ref{lem:quad-solve}, and compute the new $\Delta$ and $\pi$ values by Megiddo's algorithm. These are strongly polynomial subroutines and
return rational values of size polynomially bounded in the input. Namely, solving (\ref{sys:eq}) requires first contracting components of linear arcs and modifying costs and demands by additive terms. In the contracted instance, we need to solve a system of linear equations by exact arithmetics. This can be done by maintaining that the sizes of numbers in the output are polynomially bounded in the input size, see e.g. \cite[Chapter 3]{Schrijver}. The new $\Delta$ and $\pi$ are obtained using Megiddo's strongly polynomial parametric search algorithm. It is immediate that $\Delta$ will be of polynomial encoding size, since it equals the cost-to-time ratio of a certain cycle, with both costs and times of polynomial encoding size.

Consider now the phases between any two successful trials (or between the initialization and the first successful trial); the bound on the number of such phases is $O(\log m)$. The value of $\Delta$  decreases by a factor of 2 at the end of each phases, and the value of $f$ is modified by $\pm \Delta$ in path augmentations and by $\pm\Delta/2$ in the \textsc{Adjust} subroutine. Consequently, the flow remains an integer multiple of $\Delta$ on the arcs $ij\in E\setminus F$ up to the \textsc{Adjust} subroutine(see also  Lemma~\ref{lem:enhanced}(i)). On arcs $ij\in F$, it will be the sum of the value returned by \textsc{Trial-and-Error}, plus an integer multiple of $\Delta$. The bound $O(n+m_N)$ on the number of path augmentations, and the bound $O(\log m)$ on the number of phases guarantees that the numerators also remain polynomially bounded.
%

\subsection{Fisher's market with linear utilities}\label{sec:market}
In the {\sl linear  Fisher market model}, we are given a set $B$ of buyers and a set $G$ of goods. Buyer $i$ has a budget $m_i$, and there is one divisible unit of each good to be sold.
 For each buyer $i\in B$ and good $j\in G$, $U_{ij}\ge 0$ is the utility accrued by buyer $i$ for one unit of good $j$.
Let $n=|B|+|G|$; let $E$ be the set of pairs $(i,j)$ with $U_{ij}>0$
and let $m=|E|$. We assume that there is at least one edge in $E$ incident to every buyer and to every good.

An equilibrium solution consist of prices $p_j$ of the goods and 
allocations $x_{ij}$, so that {\em (i)} all goods are sold, {\em (ii)} all money of
the buyers is spent, and {\em (iii)} each buyer $i$ buys a best bundle of
goods, that is, goods $j$ maximizing $U_{ij}/p_j$. 

The classical convex programming formulation of this problem was given
by Eisenberg and Gale \cite{Eisenberg59}. Recently, Shmyrev
\cite{Shmyrev09} gave the following alternative formulation. The
variable $f_{ij}$ represents the money paid by buyer $i$ for product $j$.
\begin{align*}
\min \sum_{j\in G}p_j(\log p_j-1) &-\sum_{ij\in E}f_{ij}\log U_{ij} \\
\sum_{j\in G}f_{ij}&=m_i \quad\forall i\in B\\
\sum_{i\in B}f_{ij}&=p_j \quad\forall j\in G\\
f_{ij}&\ge 0 \quad\forall ij\in E
\end{align*}
Let us construct a network on node set $B\cup G\cup \{t\}$ as follows.
Add an arc $ij$ for every $ij\in E$, and an
arc $jt$ for every $j\in G$.
Set $b_i=-m_i$ for $i\in B$, $b_j=0$ for $j\in G$ and
$b_t=\sum_{i\in B}m_i$. Let all lower arc capacities be 0 and upper
arc capacities $\infty$. With $p_j$ representing the flow on arc $jt$, and $f_{ij}$ the flow on arc $ij$,
the above formulation is a minimum-cost flow problem with separable
convex objective. 
(The arc $jt$ is restricted, with extending the 
functions $p_j(\log p_j-1)$ to take value 0 in 0 and  $\infty$ on $(-\infty,0)$. All other arcs are free; indeed, they are linear.) 
In this section, the convention $p_j=f_{jt}$ shall
be used for some pseudoflow $f$ in the above problem.

Let us justify that an optimal solution gives a market
equilibrium. Let $f$ be an optimal solution that satisfies
(\ref{cond:opt}) with $\pi: B\cup G\cup \{t\}\rightarrow \mathbb{R}$. We may assume
$\pi_t=0$. $C'_{jt}(\alpha)=\log\alpha$ implies $\pi_j=-\log
p_j$. On each $ij\in E$ we have $\pi_j-\pi_i\le -\log U_{ij}$ with
equality if $f_{ij}>0$. With
$\beta_i=e^{\pi_i}$, this is equivalent to $U_{ij}/p_j\le \beta_i$, 
verifying that every buyer receives a best bundle of goods.

\medskip
Oracle~\ref{assump:oracle}(b) is a valid assumption, since the derivatives on arcs $ij$ between buyers and goods are
$-\log U_{ij}$, while on an arc $jt$ it is $\log f_{jt}$. The property (\ref{assump:lin-nonlin-igazi}) is straightforward.

Let us turn to Oracle~\ref{assump:trial}.
When the subroutine \textsc{Trial} is called, we transform $b$ to
$\hat b$ by changing the value at one node of each component $K$ of $F$. For
simplicity, let us always modify $b_t$
if $t\in K$, and on an arbitrary node for the other components. We shall verify the assumptions in Oracle~\ref{assump:trial} only
for such $\hat b$'s; the argument can easily be extended to arbitrary
$\hat b$ (although it is not necessary for the algorithm).
Let us call the component $K$ containing $t$ the {\sl large component}.

In \textsc{Trial}$(F)$, we want to find a potential $\pi:B\cup
G\cup\{t\}\rightarrow \R\cup\{\infty\}$, money allocations $f_{ij}$ for $ij\in F$, $i\in B$, $j\in
G$, and prices $p_j=f_{jt}$ for $jt\in F$ such that
\begin{align*}
\pi_j-\pi_i&=-\log U_{ij} \quad\forall ij\in F, i\in B, j\in G\\
\pi_t-\pi_j&=\log p_j\quad\forall jt\in F\\
\sum_{j\in G, ij\in F}f_{ij}&=\hat b_i \quad\forall i\in B\\
\sum_{i\in B, ij\in F} f_{ij}&=p_j\quad\forall jt\in F\\
\sum_{i\in B, ij\in F} f_{ij}&=\hat b_j\quad\forall jt\in E\setminus F
\end{align*}
We may again assume $\pi_t=0$. Let $P_j=e^{-\pi_j}$ for $j\in G$ and
$\beta_i=e^{\pi_i}$ for $i\in B$.  With this
  notation, $U_{ij}/P_j=\beta_i$ for $ij\in F$.
 If $jt\in F$, then $P_j=p_j$. 

Finding $f$ and $\pi$ can be done independently on the
different components of $F$.
For any component  different from the large one, all edges are
linear. Therefore we only need to find a feasible flow on a tree, and
independently, $P_j$ and $\beta_i$ values satisfying $U_{ij}/P_j=\beta_i$ on
arcs $ij$ in this component. Both of these can be performed in linear
 time in the number of edges in the tree. Note that  multiplying
 each $P_j$ by a constant $\alpha>0$ and dividing each $\beta_i$ by the same $\alpha$
yields another feasible solution.

Let  $T_1,\ldots,T_k$ be the  components of the
large component after deleting $t$. If $T_\ell$ contains a single good
$j$, then we  set $p_j=P_j=0$ ($\pi_j=\infty$). If $T_\ell$ is nonsingular, then $F$
restricted to $T_\ell$ forms a spanning tree. The equalities $U_{ij}/P_j=\beta_i$
uniquely define the ratio $P_j/P_{j'}$ for any $j,j'\in G\cap T_\ell$. We have
that $p_j=P_j$ and $\sum_{i\in B\cap T_\ell}m_i=\sum_{j\in G\cap T_\ell} p_j$ by the constraints
on the buyers in $B\cap T_\ell$ and goods in $G\cap T_\ell$; note that $\hat b_i=-m_i$ for all buyers in  $B\cap T_\ell$.
 Hence the prices in $T_\ell$ are uniquely
 determined. Then the edges in $F$
simply provide the allocations $f_{ij}$. All these computations can be
performed in $\rho_T(n,m)=O(m)$ time.

\medskip
For Oracle~\ref{assump:error}, we show that
\textsc{Error}$(f,F)$ can be implemented based on the
 Floyd-Warshall algorithm (see \cite[Chapter 5.6]{amo}). Let $\pi$ be the potential
 witnessing that $f$ is $(\Delta,F)$-feasible. Assuming $\pi_t=0$, and
 using again the notation $P_j=e^{-\pi_j}$ for $j\in G$ and
 $\beta_i=e^{\pi_i}$ for $i\in B$, we get 
\begin{equation}\label{eq:fisher-opt}
U_{ij}/P_j\le \beta_i\mbox{ if }i\in B, j\in G, ij\in E,\mbox{ with equality if }ji\in E^F_f.
\end{equation}
Furthermore, we have $p_j-\Delta\le P_j\le p_j+\Delta$ if $p_j>0$ and 
$P_j\le \Delta$ if $p_j=0$.

Let us now define $\gamma:G\times G\rightarrow \R$ as 
\[
\gamma_{jj'}=\max\left\{\frac{U_{ij'}}{U_{ij}}: i\in B, ji,ij'\in E^F_f\right\}.
\]
If no such $i$ exists, define
$\gamma_{jj'}=0$; let $\gamma_{jj}=1$ for every $j\in G$. 

\begin{claim}\label{cl:jo-p}
Assume we are given some $P_j$ values, $j\in G$.
There exists $\beta_i$ values ($i\in B$) satisfying (\ref{eq:fisher-opt}) if and only if 
$P_{j'}\ge P_j\gamma_{jj'}$ holds for every $j,j'\in G$.
\end{claim}
\begin{proof}
The condition is clearly necessary by the definition of $\gamma_{jj'}$. 
Conversely, if this condition holds, setting $\beta_i=\max_{j\in G} U_{ij}/P_j$ does satisfy (\ref{eq:fisher-opt}).
\end{proof}

If there is a directed cycle $C$ with $\Pi_{ab\in C}\gamma_{ab}>1$,
then $f$ cannot be $(\Delta,F)$-feasible for any $\Delta$.
Otherwise, we may compute $\tilde\gamma_{jj'}$ as the maximum of
$\Pi_{ab\in P}\gamma_{ab}$ over all directed paths $P$ in $E_f^F$ from $j$ to $j'$
(setting the value 0 again if no such path exists). This can be done
by the multiplicative version of the Floyd-Warshall algorithm in
$O(n^3)$ time (note that this is equivalent to finding all-pair
shortest paths for $-\log\gamma_{ab}$).

For $(\Delta,F)$-feasibility, we clearly need to satisfy
\[
(p_j-\Delta)\tilde\gamma_{jj'}\le P_j \tilde\gamma_{jj'}\le P_{j'}\le p_{j'}+\Delta.
\]
Let us define $\Delta$ as the smallest value satisfying all
these inequalities, that is,
\begin{equation}
\Delta=\max\left\{0,\max_{j,j'\in G}\frac{p_{j}\tilde\gamma_{jj'}-p_{j'}}{\tilde\gamma_{jj'}+1}\right\}.\label{eq:delta-def}
\end{equation}
We claim that $f$ is $(\Delta,F)$-feasible with the above choice.
For each $j\in G$, let $P_{j}=\max_{h\in
  G}\tilde\gamma_{hj}(p_h-\Delta)$. It is easy to verify that these $P$
values satisfy $P_{j'}\ge P_j\gamma_{jj'}$, and $p_j-\Delta\le P_j\le
p_j+\Delta$. The condition (\ref{eq:fisher-opt}) follows by Claim~\ref{cl:jo-p}.

The complexity of \textsc{Error}$(f,F)$ is dominated by the Floyd-Warshall
algorithm, $O(n^3)$ \cite{floyd62}. The problem is defined on an
uncapacitated network, with the number of nonlinear arcs $m_N=|G|<n$.
Thus Theorem~\ref{thm:running-time-bound} gives the following.
\begin{theorem}\label{thm:fisher-run}
For Fisher's market with linear utilities, the algorithm finds an optimal solution in
$O(n^4+n^2(m+n\log n)\log n)$. 
\end{theorem}
The algorithm of Orlin \cite{Orlin10} runs in $O(n^4\log n)$ time,
assuming  $m=O(n^2)$. Under this assumption, we get the same running
time bound.

\medskip
To prove that the algorithm is strongly polynomial, let us verify the nontrivial requirement (iii) (see the Introduction).
As discussed in Section~\ref{sec:shortest}, if the input is rational, we shall maintain that 
$f$, $\Delta$ and the $e^{\pi_i}$ values are rational; the latter are used in the computations instead of the $\pi_i$'s.
At the initialization and in every successful trial, the subroutines described above are strongly polynomial and therefore return rational $f$, $\Delta$ and  $e^{\pi_i}$ values,
of size polynomially bounded in the input (note that the  $e^{\pi_i}$ values above are denoted by $P_i$ for $i\in G$ and $\beta_i$ for $i\in B$, and $e^{\pi_t}=1$). Between two successful trials, we can use the same argument as in Section~\ref{sec:quad} for quadratic costs: there are $O(\log m)$ such iterations, $\Delta$ is  divided by two at the end of every phase,  the path augmentations change $f$ by $\pm \Delta$ and \textsc{Adjust} by $\pm\Delta/2$.
The multiplicative Dijkstra algorithm described in the Appendix also maintains rational $e^{\pi_i}$ values of polynomial encoding length. 

\subsection{Fisher's market with spending constraint utilities}
The {\sl spending constraint utility} extension of linear Fisher
markets was defined by Vazirani \cite{Vazirani10spending}. In this
model, the utility of a buyer decreases as the function of the money
spent on the good. Formally, for each pair $i$ and $j$ there is a
sequence
$U_{ij}^1>U_{ij}^2>\ldots>U_{ij}^{\ell_{ij}}> 0$ of utilities with numbers
$L_{ij}^1,\ldots,L_{ij}^{\ell_j}>0$. Buyer $i$ accrues utility
$U_{ij}^1$ for every unit of $j$ he purchased by spending the first 
$L_{ij}^1$ dollars on good $j$, $U_{ij}^2$ for spending the next $L_{ij}^2$ dollars, etc.
These $\ell_{ij}$ intervals corresponding to the pair $ij$ are called
{\sl segments}. $\ell_{ij}=0$ is allowed, but we assume $\sum_{j\in G} \ell_{ij}>0$ for all $i\in B$ and
$\sum_{i\in B} \ell_{ij}>0$ for all $j\in G$. 
Let $n=|B|+|G|$ denote the total number of buyers and goods, and $m$ denote the total number of segments. Note that $m>{n}^2$ is also possible.

No extension of the Eisenberg-Gale convex program is known to capture
this problem. 
The existence of a convex programming formulation is left as an open
question in \cite{Vazirani10spending}. This was settled by Devanur et al. \cite{Birnbaum11},  giving a convex program based on  Shmyrev's formulation.
Let $f^k_{ij}$ represent the money paid by buyer $i$ for the $k$'th
segment of product $j$, $1\le k\le \ell_{ij}$.
\begin{align*}
\min \sum_{i\in G}p_j(\log p_j-1) &-\sum_{i\in B,j\in G,1\le k\le \ell_{ij}}f^k_{ij}\log U^k_{ij} \\
\sum_{j\in G,1\le k\le \ell_{ij}}f_{ij}^k&=m_i \quad\forall i\in B\\
\sum_{i\in B, 1\le k\le \ell_{ij}}f_{ij}^k&=p_j \quad\forall j\in G\\
0\le f_{ij}^k&\le L_{ij}^k \quad\forall ij\in E.
\end{align*}

This gives a convex cost flow problem again on the node set $B\cup
G\cup \{t\}$, by adding $\ell_{ij}$ parallel arcs from $i\in B$ to
$j\in G$, and arcs $jt$ for each $j\in G$. The upper capacity on the
$k$'th segment for the pair $ij$ is $L_{ij}^k$.
To apply our method, we first need to transform it to an equivalent
problem without 
upper capacities. This is done by replacing the arc
representing the $k$'th segment of $ij$ by a new node $(ij,k)$ and
two arcs $i(ij,k)$ and $j(ij,k)$. The node demand on the new node
is set to $L_{ij}^k$, while on the good $j$, we replace the demand 0
by $-\sum_{i,k}L_{ij}^k$, the negative of the sum of capacities of all
incident segments. The cost function on $i(ij,k)$ is $-\log U_{ij}^k\alpha$, while
the cost of $j(ij,k)$ is $0$.
Let $S$ denote the set of the new $(ij,k)$ nodes. This modified graph
has $n'=n+m+1$ nodes and $m'=2m+|G|$ arcs.

 Assumption (\ref{assump:lin-nonlin-igazi}) is clearly valid. Oracle~\ref{assump:oracle}(b) is
satisfied the same way as for linear Fisher markets, using an oracle
for the $e^{C'_{ij}(\alpha)}$ values.

In \textsc{Trial}$(F)$, we want to find an $F$-tight flow $f'$ on the
extended network, witnessed by the  potential $\pi:B\cup
 S \cup G\cup\{t\}\rightarrow \R$.
We may assume $\pi_t=0$. Let $P_j=e^{-\pi_j}$ for $j\in G$ and
$\beta_i=e^{\pi_i}$ for $i\in B$ and $S_{ij}^k=e^{-\pi_{(ij,k)}}$.
For the $k$'th segment of $ij$, $U_{ij}^k/S_{ij}^k= \beta_i$ if $i(ij,k)\in F$ and
$S_{ij}^k= P_j$ if $j(ij,k)\in F$.
 
As for linear Fisher markets, if a component of $F$ does not contain
$t$, we can simply compute all potentials and flows as $F$ is a
spanning tree of linear edges in this component. 

For the component $K$ with $t\in K$, let $T_{\ell}$ be a component of
$K-t$. $F$ is a spanning tree of linear edges in $T_{\ell}$ as well,
therefore the ratio $P_j/P_{j'}$ is uniquely defined for any $j,j'\in
G\cap T_{\ell}$.
On the other hand, we must have $P_j=p_j$, and we know that
$\sum_{j\in G\cap T_{\ell}}p_j=-\sum_{v\in T_{\ell}}b_v$ by flow
conservation. These determine the $P_j=p_j$ values, and thus all other
$\beta_i$ and $S_{ij}^k$ values in the component as well. The support of
the flow $f_{ij}$ is a tree and hence it can
 also easily computed. The running time of \textsc{Trial} is again
linear, $\rho_T(n',m')=O(m')=O(m)$.

\medskip

\textsc{Error}$(f,F)$ can be implemented the same way as for the linear
Fisher market. We shall define the values $\gamma:G\times G\rightarrow
\R$ so that $P_{j'}\ge P_j\gamma_{jj'}$ must hold, and conversely,
given $P_j$ prices satisfying these conditions, we can  define the
$\beta_i$ and $S_{ij}^k$ values feasibly.
Let
\begin{align*}
\gamma_{jj'}=\max\Bigl\{&\frac{U_{ij'}^{k'}}{U_{ij}^{k}}: i\in B,\\
  &j(ij,k),(ij,k)i,i(ij',k'),(ij',k')j'\in E^F_f\Bigr\}.
\end{align*}
 Given these $\gamma_{jj'}$
values, the $\tilde \gamma_{jj'}$ values can be computed by the Floyd-Warshall
algorithm and the optimal $\Delta$ obtained by (\ref{eq:delta-def}) as
for the linear case.

Finding the $\gamma_{jj'}$ values can be done in $O(m')$ time, and the
Floyd-Warshall algorithm runs in $O(|G|^3)$. This gives
$\rho_E(n',m')=O(m'+|G|^3)=O(m+n^3)$. 
From
Theorem~\ref{thm:running-time-bound}, together with
Remark~\ref{remark:shortest}, we obtain:
\begin{theorem}\label{thm:spending-run}
For an instance of Fisher's market with spending constraint utilities
with $n=|B|+|G|$ and $m$ segments, the running time can be bounded by
$O(mn^3 +m^2(m+n\log n)\log m)$.
\end{theorem}
It can be verified that the algorithm is strongly polynomial the same way as for the linear case.

\section{Discussion}\label{sec:discussion}
We have given strongly polynomial algorithms for a class of
minimum-cost flow problems with separable convex objectives. This
gives the first strongly polynomial algorithms for quadratic convex
cost functions and for Fisher's market with spending constraint
utilities. For Fisher's market with linear utilities, we get the same
complexity as in \cite{Orlin10}. 

The bottleneck in complexity of all applications is the subroutine
\textsc{Trial}. However, the exact value of $err_F(f)$ is not needed:
a constant approximation would also yield the same complexity bounds.
Unfortunately, no such algorithm is known for the minimum cost-to-time
ratio cycle problem that would have significantly better, strongly
polynomial running time. Finding such an algorithm would immediately
improve the running time for quadratic costs.

A natural future direction could be to develop strongly polynomial
algorithms for quadratic objectives and const\-raint matrices with
bounded subdeterminants. This would be a counterpart of Tardos's result
\cite{Tardos86} for linear programs. Such an extension could be
possible by extending our techniques to the setting of Hochbaum and
Shantikumar \cite{Hochbaum90}.

The recent paper \cite{Vegh11} shows that linear Fisher market, along with several extension,
can be captured by a concave extension of the generalized flow model.
A natural question is if
there is any direct connection between the
concave generalized flow model and the convex minimum cost flow model
studied in this paper.
Despite certain similarities, no reduction is known in any
direction. Indeed, no such reduction is
 known even between the linear special cases, that is, generalized flows and minimum-cost flows.
The perfect price discrimination model 
\cite{Goel10}, and the Arrow-Debreu Nash-bargaining problem
\cite{Vazirani11}, are  instances of  the concave generalized flow
model, but they are not known to be reducible to convex cost flows. On
the other hand, the spending constraint utility model investigated in
this paper is not known to be reducible to concave generalized flows.

The algorithm in \cite{Vegh11} is not strongly
polynomial. 
Even for linear generalized flows, the first strongly polynomial algorithm was only given very recently \cite{vegh-strongly}.
One could try to extend
this to a class of concave generalized flows in a similar manner as in the current paper, i.e. assuming certain oracles.
This could lead to strongly polynomial algorithms for the market problems
that fit into this model.

A related problem is finding a strongly polynomial algorithm for
 minimizing a separable convex objective over a
submodular polyhedron. Fujishige \cite{Fujishige84} showed that for
separable convex quadratic costs, this is essentially equivalent to
 submodular function minimization.
Submodular utility allocation
markets by Jain and Vazirani \cite{Jain10} also fall into this class,
and are solvable in strongly polynomial time; see also
Nagano \cite{Nagano07}. Other strongly polynomially solvable special cases
are given by Hochbaum and Hong \cite{Hochbaum95}.

A common generalization of this problem and ours is
minimizing a separable convex objective over a submodular flow
polyhedron.
Weakly polynomial algorithms were given by Iwata \cite{iwata97} and by
Iwata, McCormick and Shigeno \cite{iwata03}. One might try to
develop strongly polynomial algorithms for 
some class of separable convex objectives; in
 particular, for separable
 convex quadratic functions.

\subsection*{Acknowledgment}
The author is grateful to an anonymous referee for several suggestions that helped to improve the presentation.
\bibliographystyle{abbrv}
\bibliography{quadratic}

\section*{Appendix}
In this Appendix we describe two variants of Dijkstra's algorithm that are used for the shortest path computations in our algorithm. This is an equivalent description of the well-known algorithm, see e.g. \cite[Chapter 4.5]{amo}.
The first, standard version is shown in Algorithm~\ref{fig:dijkstra}. We start from a cost function $c$ on a digraph $D=(V,A)$ and a potential vector
$\pi$ with $c_{ij}-\pi_j+\pi_i\ge 0$ for every arc, and two designated subsets $S$ and $T$. The set $R$ is initialized as $R=S$, and denotes in every iteration the set of nodes that can be reached from $S$ on a tight path, that is, all arcs of the path satisfying $c_{ij}-\pi_j+\pi_i=0$. Every iteration increases the potential on $V\setminus R$ until some new tight arcs enter. We terminate once $R$ contains a node in $T$; a shortest path between $S$ and $T$ can be recovered using the pointers $pred(i)$.

In our algorithm, this subroutine will be applied if Oracle~\ref{assump:oracle}(a) holds. In the $\Delta$-phase, we apply it for the digraph $E_f^F(\Delta)$ and the cost function $c_{ij}=C'_{ij}(f_{ij}+\Delta)$, and the potential $\pi$ as in the algorithm. Note that if the initial $\pi$ is rational, and all $c_{ij}$ values are rational, the algorithm terminates with a $\pi$ that is also rational. Oracle~\ref{assump:oracle}(a) guarantees that if $f_{ij}$ and $\Delta$ are rational numbers, then so is $c_{ij}$.

\medskip

Algorithm~\ref{fig:dijkstra-2} shows a multiplicative variant of the previous algorithm; they are identical  after substituting
$c_{ij}=\log \gamma_{ij}$ and $\pi_i=\log \mu_i$. This variant shall be applied under  Oracle~\ref{assump:oracle}(b).
We shall assume that every $e^{\pi_i}$ value is rational, and set $\mu_i=e^{\pi_i}$, and $\gamma_{ij}=e^{C'_{ij}(f_{ij}+\Delta)}$.
The assumption guarantees that if $f_{ij}$ and $\Delta$ are rational numbers, then so is $\gamma_{ij}$. Consequently, the rationality of the  $e^{\pi_i}$ values is maintained during the computations.

\begin{algorithm}[p]
\begin{tabbing}
xxxxx \= xxx \= xxx \= xxx \= xxx \= \kill
\> \textbf{Subroutine} \textsc{Shortest Paths}\\
\> \textbf{INPUT} A digraph $D=(V,A)$, disjoint subsets $S,T\subseteq V$, a cost function $c:A\rightarrow \R$\\
\> \> and a potential vector $\pi:V\rightarrow R$ with $c_{ij}-\pi_j+\pi_i\ge 0$ for every $ij\in A$.\\
\> \textbf{OUTPUT} A shortest path $P$ between a node in $S$ and a node in $T$ and a $\pi':V\rightarrow R$\\
\> \> with $c_{ij}-\pi'_j+\pi'_i\ge 0$ for every $ij\in A$, and equality on every arc of $P$.\\
\> $R\leftarrow S$;\\
\> \textbf{for} $i\in S$ \textbf{do} $pred(i)\leftarrow NULL$;\\
\> \textbf{while} $R\cap T=\emptyset$ \textbf{do} \\
\> \> $\alpha\leftarrow \min \{c_{ij}-\pi_j+\pi_i: ij\in A, i\in R, j\in V\setminus R\}$;\\
\> \> \textbf{for} $j\in V\setminus R$ \textbf{do} $\pi_j \leftarrow \pi_j+\alpha$;\\
\> \> $Z\leftarrow\{ j\in V\setminus R :\exists ij\in A, i\in R \mbox{ such that } c_{ij}-\pi_j+\pi_i=0\}$;\\
\> \> \textbf{for} $j\in Z$ \textbf{do}\\
\> \> \> $pred(j)\leftarrow  i\in R \mbox{ such that }\exists ij\in A: c_{ij}-\pi_j+\pi_i=0$;\\
\> \> $R\leftarrow R\cup Z$;\\
\> $\pi'\leftarrow \pi$;\\
\end{tabbing}
\caption{}\label{fig:dijkstra}
\end{algorithm}

\begin{algorithm}[p]
\begin{tabbing}
xxxxx \= xxx \= xxx \= xxx \= xxx \= \kill
\> \textbf{Subroutine} \textsc{Multiplicative Shortest Paths}\\
\> \textbf{INPUT} A digraph $D=(V,A)$, disjoint subsets $S,T\subseteq V$, a cost function $\gamma:A\rightarrow \R$\\
\> \> and a potential vector $\mu:V\rightarrow R$ with $\gamma_{ij}\frac{\mu_i}{\mu_j}\ge 1$ for every $ij\in A$.\\
\> \textbf{OUTPUT} A shortest path $P$ between a node in $S$ and a node in $T$ and a $\mu':V\rightarrow R$\\
\> \> with  $\gamma_{ij}\frac{\mu'_i}{\mu'_j}\ge 1$  for every $ij\in A$, and equality on every arc of $P$.\\
\> $R\leftarrow S$;\\
\> \textbf{for} $i\in S$ \textbf{do} $pred(i)\leftarrow NULL$;\\
\> \textbf{while} $R\cap T=\emptyset$ \textbf{do} \\
\> \> $\alpha\leftarrow \min \{\gamma_{ij}\frac{\mu_i}{\mu_j}: ij\in A, i\in R, j\in V\setminus R\}$;\\
\> \> \textbf{for} $j\in V\setminus R$ \textbf{do} $\mu_j \leftarrow \alpha\mu_j$;\\
\> \> $Z\leftarrow\{ j\in V\setminus R :\exists ij\in A, i\in R \mbox{ such that } \gamma_{ij}\frac{\mu_i}{\mu_j}=1\}$;\\
\> \> \textbf{for} $j\in Z$ \textbf{do}\\
\> \> \> $pred(j)\leftarrow  i\in R \mbox{ such that }\exists ij\in A: \gamma_{ij}\frac{\mu_i}{\mu_j}=1$;\\
\> \> $R\leftarrow R\cup Z$;\\
\> $\mu'\leftarrow \mu$;\\
\end{tabbing}
\caption{}\label{fig:dijkstra-2}
\end{algorithm}

\newpage

\subsection*{Table of notation and concepts}

\vskip 0.5cm

\begin{tabular}{l|l|l}
Notation/concept & Section & Description\\
\hline
$m_L$, $m_N$ & Sec~\ref{sec:preliminaries}, after (\ref{assump:lin-nonlin-igazi}) & number of linear/nonlinear arcs\\
$\rho_f(i)$, $Ex(f)$ & Sec~\ref{sec:preliminaries}, (\ref{def:rho}) & net flow amount in node $i$/total excess \\
$E_f$, $E_f(\Delta)$ & Sec~\ref{sec:opt}, above (\ref{cond:opt})/(\ref{eq:delta-feas}) & residual graph/ $\Delta$-residual graph\\
$E_f^F$, $E_f^F(\Delta)$ & Sec~\ref{sec:revealed}, (\ref{eq:EfF})/(\ref{def:F-delta})  & $F$-residual graph/ $(\Delta,F)$-residual graph\\
$F^*$ & Sec~\ref{sec:revealed}, (\ref{def:Fcs}) & set of arcs tight in every optimal solution\\
$err_F(f)$ & Sec~\ref{sec:assump}, (\ref{def:err}) & ``error measure''\\
\hline
free/restricted arcs & Sec~\ref{sec:preliminaries} & \\
linear/nonlinear arcs &  Sec~\ref{sec:preliminaries}, after (\ref{assump:lin-nonlin-igazi}) & \\
linear acyclic arc set & end of Sec~\ref{sec:revealed}  & \\
pseudoflow & Sec~\ref{sec:preliminaries}, above (\ref{def:rho}) & \\
$F$-pseudoflow &  Sec~\ref{sec:revealed}, above (\ref{eq:EfF}) &\\
$\Delta$-feasible & Sec~\ref{sec:opt}, (\ref{eq:delta-feas})  & \\ 
$(\Delta,F)$-feasible & Sec~\ref{sec:revealed}, (\ref{eq:delta-F-feas}) & \\
$F$-optimal & Sec~\ref{sec:revealed}, (\ref{cond:F-opt}) & \\
$F$-tight & Sec~\ref{sec:assump}, (\ref{eq:tight}) & 
\end{tabular}

\end{document}